\documentclass{aims}
\usepackage{amsmath}
  \usepackage{paralist}
  
  \usepackage{cite}
  \usepackage{verbatim}
  \usepackage{amssymb}
  \usepackage{graphics} %% add this and next lines if pictures should be in esp format
  \usepackage{epsfig} %For pictures: screened artwork should be set up with an 85 or 100 line screen
\usepackage{graphicx}  \usepackage{epstopdf}%This is to transfer .eps figure to .pdf figure; please compile your paper using PDFLeTex or PDFTeXify.

\usepackage{hyperref}

\def\currentvolume{X}
 \def\currentissue{X}
  \def\currentyear{200X}
   \def\currentmonth{XX}
    \def\ppages{X--XX}
     \def\DOI{10.3934/xx.xx.xx.xx}

\newtheorem{theorem}{Theorem}[section]
\newtheorem{corollary}{Corollary}

\newtheorem{lemma}[theorem]{Lemma}
\newtheorem{proposition}{Proposition}

\theoremstyle{definition}
\newtheorem{definition}[theorem]{Definition}
\newtheorem{remark}{Remark}
\newtheorem*{notation}{Notation}
\newtheorem{example}{Example}

\title[LP Bounds for Distributed Storage Codes] %Use the shortened version of the full title
      {Linear Programming Bounds for Distributed Storage Codes}

\author[Ali Tebbi, Terence Chan and Chi Wan Sung]{}


 \email{ali.tebbi@unisa.edu.au}
 \email{terence.chan@unisa.edu.au}
 \email{albert.sung@cityu.edu.hk}

\thanks{This work was partially supported by a grant from Australian Research
Council (DP1094571) and the University Grants Committee of the Hong Kong
Special Administrative Region, China (Project No. AoE/E-02/08).}

\begin{document}

%%%%%%%%%%%%%%
\def\N{\mathcal{N}}
\def\As{A_{\mathbf{s}}}
\def\Bs{B_{\mathbf{s}}}
\def\Cs{C_{\mathbf{s}}}
\def\SN{S_{\mathcal{\n}}}
\def\w{\mathbf{w}}
\def\z{\mathbf{z}}
\def\c{\mathbf{c}}
\def\u{\mathbf{u}}
\def\v{\mathbf{v}}
\def\h{\mathbf{h}}
\def\g{\mathbf{g}}
\def\f{\mathbf{f}}
\def\0{\mathbf{0}}
\def\s{\mathbf{s}}
\def\Aw{A_{\mathbf{w}}}
\def\Bw{B_{\mathbf{w}}}
\def\Cw{C_{\mathbf{w}}}
\def\M{\mathcal{M}}
\def\C{\mathcal{C}}
\def\I{\mathcal{I}}
\def\Cd{\mathcal{C}^\perp}
\def\FF{\mathbb{F}}
\def\FFq{\mathbb{F}_q}
\def\FFqn{\mathbb{F}_{q}^{n}}
\def\FFqk{\mathbb{F}_{q}^{k}}
\def\lambdaC{\Lambda_{\mathcal{C}}}
\def\lambdaCd{\Lambda_{\mathcal{C}^{\perp}}}

\def\calJ{{\cal J}}
\def\support{{\Lambda}}
\def\br{{\bf r}}
\def\bc{{\bf c}}
\def\bz{{\bf z}}
\def\bs{{\bf s}}
\def\bu{{\bf u}}
\def\bv{{\bf v}}
\def\bg{{\bf g}}
\def\bff{{\bf f}}
\def\ground{{\Omega}}
\def\bc{{\bf c}}
\def\solA{{\bf A}}
\def\solB{{\bf B}}
\def\solC{{\bf C}}
\def\field{{\mathbb F}_{q}}
\def\Kp{{\mathcal K}}

\def\bz{{z_{\scriptscriptstyle \N}}}
\def\by{{y_{\scriptscriptstyle \N}}}
\def\l{\left}
\def\r{\right}
\def\gf{ {\mathbb F}}
\def\rankfn{\rho}
\def\A{\mathcal A}
\def\M{\mathcal M}
\def\X{\mathcal X}
\def\B{\mathcal B}
\def\reals{\mathbb R}
\def\Z{\mathcal Z}
\def\K{\mathcal K}
\def\C{\mathcal C}
\def\p{\prime}
\def\real{{\mathbb R}}
\def\dist{{W}}

\def\x{{\bf x}}
\def\y{{\bf y}}

\def\n{{N}}
\def\k{{K}}

%%%%%%%%%%%

\maketitle

% Enter the first author's name and address:
\centerline{\scshape Ali Tebbi}
\medskip
{\footnotesize
% please put the address of the first author
 \centerline{Institute for Telecommunications Research}
   \centerline{University of South Australia}
} % Do not forget to end the {\footnotesize by the sign }

\medskip

\centerline{\scshape Terence Chan$^*$ and Chi Wan Sung$^{**}$}
\medskip
{\footnotesize
 % please put the address of the second  and third author
 \centerline{$^*$Institute for Telecommunications Research}
   \centerline{University of South Australia}
\centerline{$^{**}$Department of Electronic Engineering}
   \centerline{City University of Hong Kong}
}

\bigskip

% The name of the associate editor will be entered by an editorial staff
% "Communicated by the associate editor name" is not needed for special issue.
 \centerline{(Communicated by the associate editor name)}

%The abstract of your paper
\begin{abstract}
A major issue of locally repairable codes is their robustness. If a local repair group is not able to perform the repair process, this will result in increasing the repair cost. Therefore, it is critical for a locally repairable code to have multiple repair groups. In this paper we consider robust locally repairable coding schemes which guarantee that there exist multiple distinct (not necessarily disjoint) alternative local repair groups for any single failure such that the failed node can still be repaired locally even if some of the repair groups are not available. We use linear programming techniques to establish upper bounds on the code size of these codes. We also provide two examples of robust locally repairable codes that are optimal regarding our linear programming bound. Furthermore, we address the update efficiency problem of the distributed data storage networks. Any modification on the stored data will result in updating the content of the storage nodes. Therefore, it is essential to minimise the number of nodes which need to be updated by any change in the stored data. We characterise the update-efficient storage code properties and establish the necessary conditions of existence update-efficient locally repairable storage codes.
\end{abstract}

%%%%%%%%%%%%%%%%%%%%%%%%Section 1%%%%%%%%%%%%%%%%%%%%%%%%%%%%%%%%
\section{Introduction}

Significant increase in the internet applications such as social networks, file, and video sharing applications, results in an enormous  user data production which needs to be stored reliably. In a data storage network, data is usually stored in several data centres, which can be geographically separated. Each data centre can store a massive amount of data, by using hundreds of thousands of storage disks (or nodes). The most common issue with all storage networks is failure. The failure in storage networks is the result of the failed components which are varied from corrupted disk sectors to one or a group of failed disks or even the entire data centre due to the physical damages or natural disasters. The immediate consequence of failure in a storage network is data loss. Therefore, a repair mechanism must be in place to handle these failures. Protection against failures is usually achieved by adding redundancies \cite{macwilliams}. The simplest way to introduce redundancy is \emph{replication}. In this method multiple copies of each data block are stored in the network. In the case of any failure in the network, the content of the failed node can be recovered by its replica. Because of its simplicity in storage and data retrieval, this method is used widely in conventional storage systems (e.g., RAID-1 where each data block is duplicated and stored in two storage nodes) \cite{ReplDSN1Farsite,ReplDSN2PAST}. However, replication is not an efficient method in term of the storage requirement \cite{ReplvsEras2, ReplvsErasOcean}. 

Another method for data storage is \emph{erasure coding} which has gained a wide attention in distributed storage \cite{RAID, HybridOceanstore, ErasVSReplTotalRecall, EVENODD, MDSArray}. In this method as a generalisation of replication, a data block is divided into $\k$ fragments and then encoded to $\n$ ($\n > \k$) fragments via a group of mapping functions to store in the network \cite{ERC}. Erasure coding is capable of providing more reliable storage networks compared to replication with the same storage overhead \cite{ErasVsRepl}. It is noteworthy to mention that replication can also be considered as an erasure code (i.e., repetition code \cite{macwilliams}).

One of the challenges in storage networks using erasure codes is the repair bandwidth. When a storage node fails, the replaced node (i.e., newcomer) has to download the information content of the survived nodes to repair the lost data. In other words, the repair bandwidth of a distributed storage network is the total amount of the information which is needed to be retrieved to recover the lost data. Compared to the replication, It has been shown in \cite{ErasVsRepl} that in the systems with the same \emph{mean time to failure} (MTTF), using erasure codes reduces the storage overhead and the repair bandwidth by an order of magnitude.

From coding theory, it is well known that in a storage network employing an erasure code for data storage the minimum amount of information which has to be downloaded to repair a single failure is at least the same size of the original data \cite{ERC}. In other words, the amount of data transmitted in the network to repair a failed node is much more than the storage capacity of the node. \emph{Regenerating codes} which have been introduced in \cite{NCDSS} provide coding techniques to reduce the repair bandwidth by increasing the number of nodes which the failed node has to be connected during the repair procedure. It has been shown that there exists a tradeoff between the storage and repair bandwidth in distributed storage networks \cite{NCDSS}. Many storage codes in literature assume that lost data in a failed data centre can be regenerated from any subset (of size greater than a certain threshold) of surviving data centres. This requirement can be too restrictive and sometimes not necessary. Although these methods reduce the repair bandwidth, minimising I/O in distributed storage networks is more beneficial than minimising the repair bandwidth \cite{Khan:2011:SIR:2002218.2002224}. Storage codes such as regenerating codes that access all surviving nodes during repair process, increase I/O overhead. Disk I/O which is proportional to the number of the nodes involved in the repair process, seems to be the repair performance bottleneck of the storage networks \cite{LRCDimakis}.

The number of nodes which involve in repairing a failed node is defined as the code \emph{locality} \cite{LRCGopalan}. The concept of code locality relaxes the repairing requirement so that any failed node can be repaired by at least one set (of size smaller than a threshold) of surviving nodes. We refer to these sets as \emph{local repair groups}. Moreover, any erasure code which satisfies this property is referred as \emph{locally repairable code}.

One significant issue of locally repairable codes is their robustness. Consider a distributed storage network with $\n$ storage nodes such that there exists only one local repair group for each node in the network. Therefore, if any node fails, the survived nodes in its repair group can locally repair the failure. However, if any other node in the repair group is also not available (for example, due to the heavy loading), then the repair group cannot be able to repair the failed node anymore which will result in an increase of repair cost. To avoid this problem, it is essential to have multiple local repair groups which can be used for repairing the failure if additional nodes fail or become unavailable in the repair groups. The concept of locally repairable codes with multiple repair groups (or alternatives) was studied in \cite{LRmulti} where codes with multiple repair alternatives were also constructed. Locally repairable code construction with multiple disjoint repairing groups was investigated in \cite{LRCMultDisjoint}. Locally repairable codes with multiple erasure tolerance are introduced in \cite{6846301} where any failed node can be locally repaired even in the presence of multiple erasures. In \cite{LocComb}, locally repairable codes with multiple erasure tolerance have been constructed using combinatorial methods. A lower bound on the code size was derived and some code examples provided which are optimal with respect to this bound. A class of locally repairable codes with capability of repairing a set of failed nodes by connecting to a set of survived nodes are studied in \cite{CoopLocal}. Bounds on code dimension and also families of explicit codes were provided.

The main contributions of this paper are as follows:
\begin{itemize}
\item
We investigate the linear $(r,\Gamma,\zeta,\beta)$ robust locally repairable codes, defined in Definition \ref{def:RLRC}, where any failed node can be repaired by $\zeta$ different (i.e., distinct but not necessarily disjoint) repair groups each of size at most $r$ at the presence of any $\Gamma$ extra failures in the network. These codes have a minimum distance of $d=\beta+1$, i.e., these codes can tolerate any $\beta$ simultaneous failures.
\item
By employing the criteria from the definition of our robust locally repairable codes, we establish a linear programming bound to upper bound the code size. 
\item
We propose two robust locally repairable code examples which are optimal regarding our linear programming bound.
\item
We investigate the update efficiency of storage codes. In an erasure code, each encoded symbol is a function of the data symbols. Therefore, the update complexity of an erasure code is proportional to the number of encoded symbols that have to be changed in regard of any change in
the original data symbols. In other words, in a storage network the content of the storage nodes must be updated by any change in the stored data. We characterise the properties of update-efficient storage codes and establish the necessary conditions that the codes need to satisfy. 
\end{itemize}

Our linear programming bounds are alphabet dependant and establish upper bounds on the code size for arbitrary code parameters such as code length, minimum distance, local repair group's size, number of the local repair groups, and number of the failed nodes in the network. The significance of our linear programming bounds compared to the similar approaches (e.g., \cite{LRC-size-bound} or \cite{CombinatorialBound}) is that our bounds can be easily generalised for different variations of linear storage codes such that any criterion can be presented as a linear constraint in the problem. Moreover, the complexity of these problems can be significantly reduced by applying symmetries of the codes. 

This work is partially presented in 2013 9th International Conference on Information, Communications \& Signal Processing (ICICS) \cite{LPLR} and 2014 IEEE Information Theory Workshop (ITW) \cite{RLLC}. This work presents more detailed discussions on characterising robust locally repairable and update-efficient storage codes.   

The remaining of the paper is organised as follows. In Section \ref{Sec:background}, we review the related works on the locally repairable codes. In Section \ref{sec:LSC}, we introduce our robust locally repairable codes and propose their criteria. In Section \ref{Sec:LPbound}, we establish a linear programming problem to upper bound the code size. We also exploit symmetry on the linear programming bound to reduce its complexity. Section \ref{Sec:Update} studies the update efficiency of the locally repairable codes. We establish the necessary conditions for distributed storage codes in order to satisfy the locality and update efficiency at the same time. Examples and numerical results are provided in Section \ref{Sec:results}. The paper then concludes in Section \ref{Sec:conclu}.

%%%%%%%%%%%%%%%%%%%%%%%%Section 2%%%%%%%%%%%%%%%%%%%%%%%%%%%%%%%%
\section{Background and Related Works} \label{Sec:background}

Disk I/O which is proportional to the number of the nodes involved in the repair process, seems to be the repair performance bottleneck of the storage networks \cite{LRCDimakis}. Regenerating codes are an example of this issue where $N-1$ surviving nodes are involved in repair process of a failed node in order to minimise the repair bandwidth \cite{Khan:2011:SIR:2002218.2002224}.
Therefore, the number of nodes which help to repair a failed node is a major metric to measure the repair cost in storage networks. Locally repairable codes aim to optimise this metric (i.e., the number of helper nodes needed to repair a failed node). The locality of a code is defined as the number of code coordinates $r$ such that any code coordinate is a function of them \cite{LRCGopalan}. Consequently, in a storage network employing a locally repairable code, any failed node can be repaired using only $r$ survived nodes. In the case of the linear locally repairable storage codes which are the topic of this paper, the content of any storage node is a linear combination of the contents of a group of $r$ other nodes. Usually two types of code locality are defined in the literature as information symbol locality and all symbol locality which means that just the information symbols in a systematic code have locality $r$ and all symbols have locality $r$, respectively. In this paper, the code locality is referred to all symbol locality since we do not impose any systematic code constraint on the locally repairable storage codes. 

The problem of storage codes with small locality has been initially proposed in \cite{Pyramid,SelfRepair,SimpleRegCode}, independently. The Pyramid codes introduced in \cite{Pyramid} are an example of initial works on practical locally repairable storage codes. A basic Pyramid code is constructed using an existing erasure code. An MDS erasure code will be an excellent choice in this case, since they are more tolerable against simultaneous failures. To construct a Pyramid code using an MDS code, it is sufficient to split $\k$ data blocks into $\frac{\k}{r}$ disjoint groups. Then, a parity symbol is calculated for each group as a local parity. In other words, one of the parity symbols from the original code is split into $\frac{\k}{r}$ blocks and the other parity symbols are kept as global parities. Therefore, any single failure in data nodes can be repaired by $r$ survived nodes in the same group instead of retrieving data from $\k$ nodes. Generally, for a Pyramid code with data locality $r$ constructed from an $(\n,\k,d)$ MDS code, the code length (i.e., number of storage nodes) is

\begin{align} \label{eq:data_locality_length}
\n = \k\left(1+\frac{1}{r}\right)+d-2.
\end{align}
Obviously, low locality is gained in the cost of $(\frac{\k}{r}-1)$ extra storage. To obtain all symbol locality for the code, it is sufficient to add a local parity to each group of the global parities of size $r$ where will result in a code length of 

\begin{align} \label{eq:all_locality_length}
\n  = \left(\k+d-2\right)\left(1+\frac{1}{r}\right).
\end{align}
Both of the code lengths in \eqref{eq:data_locality_length} and \eqref{eq:all_locality_length} show a linear increase in storage overhead. Note that in general, adding a local parity to each group will decrease the minimum distance of the code which is a major challenge in constructing optimal all-symbol locally repairable codes \cite{OptimalLocal}.

The pioneer work of Gopalan $\emph{et al.}$ in \cite{LRCGopalan} is an in-depth study of codes locality which results in presenting a trade-off between the code length $\n$, data block size $\k$, minimum distance of code $d$, and code locality $r$. It has been shown that for any $[\n,\k,d]_q$ linear code with information locality $r$,   

\begin{align} \label{eq:Gopalan}
\n-\k \geq \left\lceil{\frac{\k}{r}} \right\rceil + d - 2.
\end{align}
In other words, to have a linear storage code with locality $r$, we need a linear overhead of at least $\left\lceil{\frac{\k}{r}} \right\rceil + d - 2$. As a special case of locality $r = \k$, the bound in \eqref{eq:Gopalan} is reduced to 

\begin{align}
\n-\k \geq d - 1,
\end{align}
which is the Singleton bound and is achievable by MDS codes such as Reed-Solomon codes. According to \eqref{eq:data_locality_length} and \eqref{eq:all_locality_length}, Pyramid codes can achieve the bound in \eqref{eq:Gopalan} which shows their optimality. 

Due to their importance in minimising the number of helper nodes to repair a single failure, locally repairable codes have received a fair amount of attention in literature. Different coding schemes to design explicit optimum locally repairable codes has been proposed in \cite{LocalMatroid,6655894,OptimalLocal,6478813,6846332,8006478,martnez2018universal}.

For the bound in \eqref{eq:Gopalan}, the assumption is that any storage node has entropy equal to $\alpha = \frac{\M}{\k}$, where $\M$ is the data block size to be stored. It is also mentioned that low locality $r << \k$ affects minimum distance $d$ regarding the term $\frac{\k}{r}$. It has been shown in \cite{LRCDimakis} that there exist erasure codes which have high distance and low locality if storage nodes store an additional information by a factor of $\frac{1}{r}$. The minimum distance of these codes is upper bounded as

\begin{align} \label{eq:VecLRC}
d \leq \n - \left\lceil{\frac{\k}{1+\epsilon}}\right\rceil - \left\lceil{\frac{\k}{r\left(1+\epsilon \right)}}\right\rceil + 2,
\end{align}  
where $\epsilon > 0$. The obtained LRCs in \cite{LRCDimakis} are vector codes where each data and coded symbol is represented as a vector (not necessarily of the same length). The definition of the vector LRCs, together with the bound in \eqref{eq:VecLRC}, are extended to $(r,\delta)$ localities in \cite{6655894} which we will describe later.

The main issue with locally repairable codes is that they mostly tend to repair a single node failure. This means that if any of the helper nodes in the repair group of a failed node is not available, the local repair process will fail. Therefore, it is critical to have multiple repair groups for each node in the storage network. The concept of locally repairable codes with multiple repair alternatives has been studied initially in \cite{LRmulti}. The main idea in \cite{LRmulti} is to obtain the parity check matrix of the code from a partial geometries. However, there are only a few such explicit codes due to the limited number of known constructions of partial geometries with desirable parameters. A lower and upper bound is also obtained for the rate of these codes with different locality and repair alternativity. Rawat \emph{et al.} investigated the availability in storage networks employing locally repairable codes \cite{LRCMultDisjoint}. The notion of availability is useful for applications involving hot data since it gives the ability of accessing to a certain information by different users simultaneously. Therefore, each information node must have $t$ disjoint repair groups of size at most $r$ (i.e., $(r,t)$-Information local code). An upper bound for the code minimum distance has been obtained in \cite{LRCMultDisjoint} as

\begin{align} \label{eq:availability}
d_\text{min} \leq \n-\k+1- \left(\left\lceil{\frac{\k t}{r}} \right\rceil -t \right),
\end{align}
which is a generalised version of the bound in \eqref{eq:Gopalan}. However, this bound only stands for information symbols availability with the assumption of only one parity symbol for each repair group. Different bounds on the rate and minimum distance of storage codes with locality $r$ and availability $t$ have been obtained independently in \cite{TamotBound} as

\begin{align}
\frac{\k}{\n} \leq \frac{1}{\prod_{j=1}^{t}{\left(1+\frac{1}{jr}\right)}}
\end{align} 
and

\begin{align} \label{eq:Tamo_d_min}
d \leq \n - \sum_{i=0}^{t}{\left\lfloor\frac{\k-1}{r^i}\right\rfloor}.
\end{align}
The minimum distance upper bound in \eqref{eq:Tamo_d_min} reduces to \eqref{eq:Gopalan} for $t=1$. However, the bounds in \eqref{eq:availability} and \eqref{eq:Tamo_d_min} seem to be incomparable since the bound in \eqref{eq:Tamo_d_min} is tighter than \eqref{eq:availability} in some cases and vice versa.

An upper bound on the codebook size of locally repairable codes which is dependent on the alphabet size has been obtained in \cite{MazumdarBound} as

\begin{align}
\k \leq \min_{t \in \mathbb{Z}_+}\left[tr+\k_{\text{opt}}^{(q)}\left(n-t(r+1),d\right)\right]
\end{align}
where, $\k_{\text{opt}}^{(q)}$ is the maximum dimension of a linear or non-linear code of length $\n$ and minimum distance $d$ with alphabet size of $q$. This bound is tighter than the bound in \eqref{eq:Gopalan}.

Locally repairable codes with small availability $t=2,3$ are investigated in \cite{LRCsmallAvail}. The smallest the availability the highest the code rate would be since the dependency among the code symbols is less while the code is still comparable to the practical replication codes. The rate optimal binary $(r,t)$-availability codes are designed in \cite{LRC2Erasure}.

\emph{$(r,\delta)$-locally repairable codes} are introduced in \cite{6846301} where the failed coded symbols can be repaired locally by at most $r$ symbols at the presence of multiple failures in the network. The symbol $c_i$, $1 \leq i \leq N$, of a $(N,K)$ code $\C$ with minimum distance $d$ is said to have $(r,\delta)$ locality if there exists a subcode of $\C$ with support containing $i$ whose length is at most $r+\delta-1$ and whose minimum distance is at least $\delta$. For a systematic code with all information symbols $(r,\delta)$ locality, any systematic symbol can be repaired by other $r$ symbols  even if there are $\delta-2$ extra failures. The minimum distance of these codes is upper bounded as

\begin{align} \label{eq:rdeltaLRC}
d \leq \n-\k+1-\left(\left \lceil \frac{\k}{r} \right \rceil-1\right) (\delta-1).
\end{align}
Note that \eqref{eq:rdeltaLRC} is reduced to \eqref{eq:Gopalan} for $\delta=2$. 

\emph{Partial MDS (PMDS)} codes are introduced in \cite{6478813}. A $[mN,m(N-r)-s]$ linear code is a  $(r,s)$-PMDS code where each row of the $m \times N$ array belongs to an $[N,N-r,r+1]$ MDS code. PMDS codes combine local correction of $r$ errors in each row with global correction of $s$ erasures anywhere in the array. As mentioned earlier, in classes of locally repairable codes (e.g., PMDS codes or Pyramid codes) there are local parity symbols which depend only on a specific group of data symbols and some global parities which depend on all data symbols. While the local parities provide a fast repair of single failures, the global parities provide a high resilience against simultaneous failures. The codes described as above are called \emph{Maximally Recoverable (MR)} \cite{Pyramid} if they can repair all failure patterns which are information theoretically repairable. Families of explicit maximally recoverable codes with locality are proposed in \cite{6846332}. PMDS codes and MR codes with locality are a class of $(r,\delta)$-locally repairable codes that not only have optimal minimum distance, but can correct all information theoretically correctable erasure patterns for the given $(r, \delta)$, code length $N$ and dimension $K$.

A challenge to construct these type of optimal LRCs (e.g., PMDS codes) is to obtain small field sizes. An exponential field size of $O(N^{mN})$ is obtained in \cite{7740918} for PMDS codes with arbitrary $r$ and $s$. The authors in \cite{8006478}, extend the work in \cite{6846332} for $r>1$ to obtain a relatively small filed size of $O(\max \{m,N^{r+s}\}^s)$. A further reduction on the field size compared to \cite{8006478} for the case $r \leq s$ is obtained in \cite{martnez2018universal}.

Locally repairable codes with multiple localities are introduced in \cite{MultiLocal} and \cite{UneqLocal} such that coded symbols have different localities. An upper bound on the minimum distance of these codes with multiple locality of the information symbols is derived as
\[
d \leq \n-\k-\sum_{j=1}^r \left\lceil \frac{\k_j}{j} \right \rceil +2,
\]
where, $\k_j$ is the number of the information symbols with locality $j$. Locally repairable codes with multiple $(r_i,\delta_i)$ which are proposed in \cite{Multi_r_delta_local}, generalise the $(r,\delta)$ locally repairable codes in \cite{6846301} to the locally repairable codes with multiple localities in \cite{MultiLocal} and \cite{UneqLocal}. These codes ensure different locality for coded symbols at the presence of different number of extra failures. A Singleton-like bound on the minimum distance has been derived and explicit codes have been proposed. A family of maximally recoverable LRCs with multiple localities that admit efficient updates of the parameters $(r_i,\delta_i)$, in order to adapt dynamically to different erasure probabilities, network topologies or different hot and cold data, has been proposed in \cite{martnez2018universal}.

\begin{notation}
Let $\N=\left \{1,2,3, \ldots,\n\right \}$ and $2^\N$ be the set of all possible subsets of $\N$. In this paper we represent a subset $\nu$ of $\N$ as a binary vector $\w=[w_1,\ldots,w_N]$ such that $w_i=1$ if and only if $i\in{\N}$. In our convention, $\nu$ and $\w$ can be used interchangeably. Similarly, we use zero vector $\mathbf{0}$ to present an empty set $\varnothing$. Also, we will make no distinction between an index $i$ and $\{i\}$ for a single element set.    
\end{notation}

%%%%%%%%%%%%%%%%%%%%%%%%Section 3%%%%%%%%%%%%%%%%%%%%%%%%%%%%%%%%
\section{Linear Locally Repairable Codes} \label{sec:LSC}

\subsection{Linear Storage Codes}

Let $ \FFq^\n=\left \{[z_1,\ldots,z_\n]:z_i \in \FFq \text{~for all~} i \in \N \right \}$ be the vector space of all $\n$-tuples over the field $\FFq$. A linear storage code $\C$ generated by a $\k \times \n$ full rank (i.e. rank $\k$) generator matrix $G$  is a subspace of $\FFq^\n$ over the field $\FFq$. Here, $\k$ is the dimension of the subspace. The linear code $\C=\left \{\u G : \u \in \FFq^\k \right \}$ is defined as the set of all linear combinations (i.e. spans) of the row vectors of the generator matrix $G$.

To store a data block of size $\M$ in a storage network using a linear code $\C$, the original data is first divided into $\k$ data blocks of size $ {\M}/{\k}$. We may assume without loss of generality that the original data is a matrix with $ {\M}/{\k}$ rows and $\k$ columns. 
Then  each block (or row) will be encoded using the same code. In this paper, we may assume without loss of generality that  there is only one row (and thus $\M = \k$).
Let $\u = [u_1,\ldots,u_\k]$ be the block of information to be stored. Then the storage node $i$ will store the symbol $\u \g_{i}$ where $\g_{i}$ is the $i^{th}$ column of the generator matrix $G$. 
Notice that if $X_{i}$ is the content that the storage node $i$ stores, then 
$(X_{1}, \ldots, X_{\n})  $   is a codeword in $\C$. 
To retrieve the stored information, it is sufficient to retrieve the information content of $\k$ nodes   (which are indexed by $\I$)  such that 
the columns $\g_{i} $ for $\forall i\in \I$ are all independent. Specifically, 

\begin{align}
\u=\c_\I G_{\I}^{-1},
\end{align}
where   $\c_\I$ is the row vector whose $i^{th}$ entry equals to $\u \g_{i}$ and $G_\I$ is a submatrix of the generator matrix $G$ formed by the columns $\g_i$ for $\forall i \in \I$. 

\begin{definition}[Support] \label{def:support}
For any vector $\c=\left[c_1,\ldots,c_\n\right] \in \FFq^\n$, its  support $\lambda(\c)$   is defined as a subset of $\N=\left\{1,\ldots,\n \right\}$ such that $i \in \lambda(\c)$ if and only if $c_i \neq 0$.
\end{definition}

For a linear storage code $\C$ over $\FFq$, any codeword $\c$ is a vector of $\n$ symbols $[c_1, \ldots, c_\n]$ where $c_i \in \FFq$ for all $i = 1, \ldots, \n$. Since we assume that code $\C$ is linear, for any codeword $\c \in \C$ the vector $a\c$ for all $a \in \FFq$ and $a \neq 0$ is also a codeword in $\C$. Therefore, there exists at least $q-1$ codewords in $\C$ with the same support (excluding the zero codeword). 

\begin{definition}[Support Enumerator]
The support enumerator $\lambdaC(\w)$ of the code $\C$ is the number of codewords in $\C$ with the same support $\w$, i.e.,  
\[
\lambdaC(\w) \triangleq \left | \left \{\c \in \C:\lambda(\c)=\w \right \} \right |, \qquad \forall \w \subseteq \N.
\]
\end{definition}

The support enumerator of a code is very useful to characterise the relation between a code $\C$ and its dual code $\Cd$ via MacWilliams theorem, hence we aim to bound the size of the locally repairable codes using the dual code properties. Any  linear code $\C$ can be specified by a $\left (\n-\k \right) \times \n$ parity check matrix $H$ which is orthogonal to the generator matrix $G$ such that
\[
GH^\top = \0,
\]
where $\0$ is a $\k \times (\n-\k)$ zero matrix. Then, this parity check matrix $H$ can generate another linear code $\Cd$ called dual code. It is well known that 
\[
\Cd=\left \{\h \in \FFq^\n : \h \c^\top=0 \text{~for all~} \c \in \C \right \}.
\] 
In other words, a linear code $\C$ can be uniquely determined by its dual code $\Cd$. The number of the codewords in $\Cd$ for any support $\w$ (i.e., support enumerator $\lambdaCd(\w$))   can also be determined uniquely by the support enumerator $\lambdaC(\w)$ of the code $\C$ via the well-known MacWilliams identities.

There are various variations or extensions of MacWilliams identity, all based on association schemes. In the following, we will restate the MacWilliam identity in favour of our purpose in this paper. 

\begin{proposition}[MacWilliams identity~{\cite[Ch.~5, Theorem 1]{macwilliams}}] \label{prop:duality}
Let $\C$ be an $(\n,\k)$ linear code and  $\Cd$ be its dual. Then for any codeword support $\w=\left[w_1,\ldots,w_\n \right] \subseteq \N$, 

\begin{align}
\lambdaCd(\w)=\frac{1}{\left | \C \right |} \sum_{\s=\left[s_1,\ldots,s_\n\right] \in 2^\N} \lambdaC(\s)\prod_{j=1}^{\n}\kappa_q(s_j,w_j)\geq0 \label{eq:duality}
\end{align}
where

\begin{align}\label{eq:prop1}
\kappa_q \left(s,w\right)=
\begin{cases}
 1 & \text{ if } w=0 \\ 
 q-1 & \text{ if } s=0  \text{ and }  w=1 \\ 
 -1 & \text{ otherwise. } 
\end{cases}  
\end{align}
\end{proposition}

\begin{remark}
Here, subsets  $\w$ and $\s$ of $\N$   are considered as binary vectors with nonzero coordinates corresponding to the indices indexed in them. 
%For instance, $\s \cap \w$ means the coordinates that both $\s$ and $\w$ are nonzero.Therefore, equation \eqref{eq:DualCodeEnumFunc} can be rewritten as \eqref{eq:duality}.
\end{remark}

In the next subsection we will show the significance of dual code support enumerator in our definition of robust locally repairable codes.
%%%%%%%%%%%%
\subsection{Locality of Linear Storage Codes}

A linear storage code has locality $r$, if for any failed node $i$, there exists a group of survived nodes of size at most $r$ (i.e., repair group) such that the failed node can be repaired by this group of nodes~\cite{LRCGopalan}.  To be precise,  consider a linear code $\C$ and   its dual code $\Cd$. 
Let $\h=\left[h_1,\ldots,h_\n \right]$ and $\c=\left[c_1,\ldots,c_\n \right]$ be the codewords in $\Cd$ and $\C$, respectively. Then, from the definition of the dual code 

\begin{align} \label{eq:h_c_innerprod1}
\h \c^\top=0, \quad \forall  \h \in \Cd, \:  \c \in \C.
\end{align}

%As we mentioned earlier, $(X_{1},\ldots,X_{n})$ represents the information content of the storage nodes which is a codeword from our storage code $\C$. Then, from \eqref{eq:h_c_innerprod1}
%\begin{align} \label{eq:h_c_innerprod2}
%\sum_{i\in\N} h_{i}X_{i} = 0.
%\end{align}
%However, the zero coordinates of $\h$ do not play any role in the equation \eqref{eq:h_c_innerprod2}. Then, it can be rewritten as
Consequently, 

\begin{align}
\sum_{j \in \lambda(\h)}h_j X_j = 0,
\end{align}
where $\lambda(\h)$ is the support of the dual codeword $\h$. 
%As such,   any codeword $\h$ in the dual code $\Cd$ determines a repair group for the nodes indicated by the index of the nonzero coordinates of $\h$ (i.e., support of $\h$).
%
Now, if the node $i$ fails and $i \in \lambda(\h)$, the content of the failed node can be recovered from the set of surviving nodes indexed by $\lambda(\h) \setminus i$. In particular, 
%\begin{align}\nonumber
$X_i=-h_{i}^{-1}\sum_{j \in \lambda \left(h\right)\setminus i}X_jh_j.$
%\end{align} 
As such, the set $\lambda(\h)\setminus i$ can be seen as  a repair group for node $i$.

\begin{definition}[Locally repairable code~\cite{LRCDimakis}]
An $(r,\beta)$ locally repairable code $\C $ is a linear code which  satisfies 

\begin{enumerate}
\item
\textbf{Local Recovery (LR).}
for any $i \in \N$, there exists $\h \in \Cd$ such that $i \in \lambda(\h)$  and $|\lambda(\h)|-1 \le r$.
\item
\textbf{Global Recovery (GR).} $\lambdaC(\w)=0$, for all $\w \subseteq \N$ such that $1 \leq |\w| \leq \beta$.
\end{enumerate} 
\end{definition}

For a $(r,\beta)$ locally repairable code, the local recovery criterion implies that any failed node $i$, can always be repaired by retrieving the contents stored in at most $r$ nodes. On the other hand, the global recovery criterion guarantees that the storage network can be regenerated as long as there are at most $\beta$ simultaneous node failures. In other words, minimum distance of the code is $d=\beta+1$.

\begin{remark} 
The LR criterion  requires the existence of a set of $r$ surviving nodes which can be used to repair a failed node. However, there is no guarantee that a failed node can still be efficiently repaired\footnote{
Strictly speaking, by the GR criterion, the failed node can still be repaired if there are no more than $\beta$ node failures. However, the global recovery cannot guarantee that each failed node can be repaired efficiently. In other words, one may need a much larger set of surviving nodes to repair one failed node.  
} if a few additional nodes also fail (or become unavailable). Therefore, it is essential to have alternative repair groups in case of multiple node failures.
\end{remark}

\begin{definition}[Robust locally repairable code] \label{def:RLRC}
An $(r,\beta,\Gamma,\zeta)$ robust locally repairable code is a linear code satisfying  the following criteria:
\begin{enumerate}
\item
\textbf{Robust Local Recovery (RLR).}
for any $i \in \N$ and $\gamma \subset \N \setminus i$ such that $|\gamma| = \Gamma$, there exists $\h_{1}, \ldots , \h_{\zeta} \in \Cd$ such that for all $j=1, \ldots, \zeta$, 
\begin{enumerate}
\item 
$i \in \lambda(\h_{j})$, $\gamma \cap \lambda(\h_{j}) = \emptyset$, and $|\lambda(\h_{j})|-1 \le r$.

\item 
$\lambda(\h_{j}) \neq \lambda(\h_{\ell})$ for $\ell \neq j$.

\end{enumerate}

\item
\textbf{Global Recovery (GR).} $\lambdaC(\w)=0$, for all $\w \subseteq \N$ such that $1 \leq |\w| \leq \beta$.
\end{enumerate} 
\end{definition}

The RLR criterion  guarantees that a failed node can be repaired locally from any one of the $\zeta$ groups of surviving nodes of size $r$ even if $\Gamma$ extra nodes fail.  In the special case when  $\Gamma=0$, then the robust locally repairable codes are reduced to locally repairable codes with multiple repair alternatives as in \cite{LRmulti}. When $\zeta=1$, then it further reduces to the traditional locally repairable codes. 
The GR criterion is the same as that for locally repairable codes.

%%%%%%%%%%%%%%%%%%%%%%%%Section 4%%%%%%%%%%%%%%%%%%%%%%%%%%%%%%%%
\section{Linear Programming Bounds} \label{Sec:LPbound}

One of the most fundamental problems in storage network design is to optimise the tradeoff between the costs for storage and repair. For a storage network using an $(\n,\k)$ linear code the capacity of each node is $\frac{\M}{\k}$ where $\M$ is the information size. In other words, the storage cost (per information per node) is given by $\frac{\M}{\log_q |\C|}$. Since our coding scheme is linear locally repairable, a group of $r$ nodes participate in repairing a failed node. Thus, the repair cost of a single failure is given by $\frac{r\M}{\log_q |\C|}$. Consequently, by maximising the codebook size $|\C|$ the cost will be minimised.

In this section, we will obtain an upper bound for the maximal codebook size, subject to robust local recovery and global recovery criteria in Definition \ref{def:RLRC}.  

\begin{theorem} \label{thm:NLOP}
Consider any $(r,\beta,\Gamma,\zeta)$ robust locally repairable code $\C$. Then, $|\C|$ is upper bounded by the optimal value of the following optimisation problem.
\begin{equation}
\begin{array}{lll}
\textbf{maximize} \quad \sum_{\w \subseteq \N}\Aw \\ 
\textbf{subject to}  \\ 
\qquad \Aw \geq 0 & \quad \forall \w \subseteq \N \\ 
\qquad \Bw = \frac{\sum_{\s \subseteq \N} \As \prod_{j=1}^{n} \kappa_q (s_j,w_j)}{\sum_{\s \subseteq \N} \As} & \quad \forall \w \subseteq \N \\ 
\qquad \Bw  \geq 0 & \quad \forall \w \subseteq \N \\ 
\qquad \Aw = 0 & \quad 1 \leq |\w| \leq \beta \\ 
\qquad A_{\varnothing} = 1 & \\ 
\qquad \displaystyle \sum_{\s \in \Omega_{i}: \gamma \cap \s = \emptyset } \Bs  \geq  \zeta (q-1) & \quad \forall i \in \N, \gamma \in \Delta_{i}
\end{array} \label{Eq:thm1}
%\qquad \displaystyle \sum_{\s : i \in \s, \gamma \cap \s = \emptyset, \text{ and } |\s|-1 \leq r} \Bs  \geq  \zeta (q-1) & \quad \forall i \in \N, \gamma \in \Delta_{i}
%\end{array} \label{Eq:thm1}
\end{equation}
where 
$\Omega_{i}$ is the collection of all subsets of $\N$ that contains $i$ and of size at most $r+1$ and 
$\Delta_{i}$ is the collection of all subsets of $\N \setminus i$ of size at most $\Gamma$.
\end{theorem}

\begin{proof}
Let $\C$ be a $(r,\beta,\Gamma,\zeta)$ robust locally repairable code. Then, we define

\begin{align} \nonumber
\Aw & \triangleq \lambdaC(\w), \quad \w \subseteq \N \\ \nonumber
\Bw & \triangleq \lambdaCd(\w), \quad \w \subseteq \N.
\end{align}
The objective function of the optimisation problem in \eqref{Eq:thm1} is the sum of the number of all  codewords with support $\w$, which is clearly equal to the size of the code $\C$.   Since $\Aw$ and $\Bw$ are enumerator functions of the code $\C$ and $\Cd$  respectively,  $\Aw, \Bw \geq 0$ for all $\w \subseteq \N$. The constraint $A_{\varnothing}=1$ follows from the fact that the zero codeword is contained in any linear code.   The second constraint follows from \eqref{eq:duality} by substituting   $|\C|$ and $\lambdaC(\s)$ with $\sum_{\s \subseteq \N}\As$ and $\As$, respectively. The fourth constraint follows directly from GR criterion which says that there are no codewords with support of size $\leq \beta$ (i.e., the minimum distance of code is $d = \beta +1$). Finally, the last constraint follows from the RLR criterion. Since the repair groups are specified by dual code $\Cd$, this constraint bounds the number of codewords in the dual code with support of size at most $r+1$. The support of the corresponding codewords must contain the index of the failed node $i$ chosen from set $\Omega_i$. However, none of the extra failed nodes indexed in $\gamma$ can be included in these supports. As mentioned earlier, due to the fact that $\C$ is a linear code, for any nonzero $\h \in \Cd$, then $a\h \in \Cd$ for all nonzero $a \in \FFq$. Let $\w = \lambda(\h)$. Then, $B_{\w} \ge (q-1)$ (i.e., there exist at least $q-1$ codewords in dual code with support $\w$).
Consequently, 
\[
\sum_{\s \in\Omega_{i} :  \gamma \cap \s =\emptyset } \Bw \geq \zeta(q-1).
\]
In other words, this constraint guarantees that there exist at least $\zeta$ repair groups to locally repair a failed node at the presence of $|\gamma| = \Gamma$ extra failures. 
Therefore, $(\Aw, \Bw : \w \subseteq \N)$ satisfies the constraints in the maximisation problem in \eqref{Eq:thm1} and  the theorem follows.      
\end{proof}

The optimisation problem in Theorem \ref{thm:NLOP} derives an upper bound on the codebook size $|\C|$ (i.e., Code dimension $K$) of a robust linear locally repairable code with minimum distance of $d = \beta +1$ over $\FFq$ such that any failed node can be repaired by $\zeta$ repair groups of size $r$ at the presence of any $\Gamma$ extra failure.  Although the maximisation problem in \eqref{Eq:thm1} is not a linear programming problem, it can be converted to one easily with a slight manipulation as follows: 
%\begin{equation}
%\begin{array}{lll}
%\textbf{maximize} \quad \sum_{\w \subseteq \N}\Aw \\ 
%\textbf{subject to}  \\ 
%\qquad \Aw \geq 0 & \quad \forall \w \subseteq \N \\ 
%\qquad \Cw = \sum_{\s \subseteq \N} \As \prod_{j=1}^{n} \kappa_q (s_j,w_j) & \quad \forall \w \subseteq \N \\ 
%\qquad \Cw  \geq 0 & \quad \forall \w \subseteq \N \\ 
%\qquad \Aw = 0 & \quad 1 \leq |\w| \leq \beta \\ 
%\qquad A_{\varnothing} = 1 & \\ 
%\qquad \displaystyle \sum_{\s\in\Omega_{i} : \gamma \cap \s =\emptyset } \Cs  \geq \zeta(q-1) \sum_{\w \subseteq \N}\Aw & \quad \forall i \in \N, \gamma\in\Delta_{i}. 
%\end{array}
%\end{equation}
%The problem can be further simplified to
\begin{equation} \label{eq:LPBNS}
\begin{array}{lll}
\textbf{maximize} \quad \sum_{\w \subseteq \N}\Aw \\ 
\textbf{subject to}  \\ 
\qquad \Aw \geq 0 & \quad \forall \w \subseteq \N \\ 
\qquad \sum_{\s \subseteq \N} \As \prod_{j=1}^{n} \kappa_q (s_j,w_j) \geq 0 & \quad \forall \w \subseteq \N \\ 
\qquad \Aw = 0 & \quad 1 \leq |\w| \leq \beta \\ 
\qquad A_{\varnothing} = 1 & \\ 
\qquad \displaystyle \sum_{\w \in \Omega_{i} : \gamma \cap \w =\emptyset} \left(\sum_{\s \subseteq \N} \As \prod_{j=1}^{n} \kappa_q (s_j,w_j) \right)  \\
\hspace{2cm} \geq \zeta(q-1) \sum_{\w \subseteq \N}\Aw   &  \forall i \in \N, \gamma\in\Delta_{i}. 
%\qquad \displaystyle \sum_{\w : i \in \w, \gamma \not\subset \w, \text{ and } |\w|-1 \leq r} (\sum_{\s \subseteq \N} \As \prod_{j=1}^{n} \kappa (s_j,w_j))  \\
%\geq (q-1) \sum_{\w \subseteq \N}\Aw & \quad \forall i \in \N. 
\end{array} 
\end{equation}

%The Linear programming problem in \eqref{eq:LPBNS}

The complexity of the linear programming problem in \eqref{eq:LPBNS} will increase exponentially with the number of storage nodes $\n$. In the following, we will reduce the complexity of the Linear Programming problem in \eqref{eq:LPBNS} by exploiting the symmetries.% \cite{SymmetryDSN}. 

To this end, suppose a $(r,\beta,\Gamma,\zeta)$ linear storage code $\C$ is generated by the matrix $G$ with columns $\g_i, ~ i=1, 2, \ldots , \n$. Let $S_\N$ be the symmetric group on $\N$ whose elements are all the permutations of the elements in $\N$ which are treated as bijective functions from the set of symbols to itself. Clearly, $|S_\N|=\n!$. Let $\sigma$ be a permutation on $\N$ (i.e.,  $\sigma \in S_\N$). 
Together with the code $\C$, each permutation $\sigma$ defines a new code $\C^\sigma$ specified with the generator matrix columns $\f_i,~i=1, 2, \ldots, \n$, such that for all $i \in \N$, $\f_i=\g_{\sigma(i)}$. Since all the codewords are the linear combinations of the generator matrix rows and the permutation $\sigma$ just changes the generator matrix columns position, the permutation cannot affect the minimum distance of the code (i.e., every codeword $\c^\sigma \in \C^\sigma$ is just a permuted version of the corresponding codeword $\c \in \C$). Therefore, the code $\C^\sigma$ is still a $(r,\beta,\Gamma,\zeta)$ linear storage code and we have the following proposition.

\begin{proposition} \label{prop:2}
Suppose $(a_\w : \w \subseteq \N)$ satisfies the constraint  in the optimisation problem  \eqref{eq:LPBNS}. For any   $\sigma \in S_\N$, let 
\[
(a_{\w}^{\sigma} : \w \subseteq \N)=a_{\w}^{\sigma} = a_{\sigma(\w)},
\] 
where $\sigma(\w) \triangleq \{ \sigma(i) : i \in \w \}$. Then $(a_{\w}^{\sigma} : \w \subseteq \N)$ also satisfies the constraint in  \eqref{eq:LPBNS}.
\end{proposition} 

%The Proposition \ref{prop:2} follows from the fact that $a_{\w}^{\sigma}$ can be considered as an enumerator function corresponds to the permuted version of the code $\C$ (i.e., $\C^\sigma$). Since the code $\C$ is a $(r,\beta,\Gamma,\zeta)$ code, the code $\C^\sigma$ is also a $(r,\beta,\Gamma,\zeta)$ code. Therefore, its enumerator function $a_{\w}^{\sigma}$ will be feasible in \eqref{Eq:linprog1}.

\begin{corollary}
Let
\[
a_{\w}^{*} = \frac{1}{|S_\N|} \sum_{\sigma \in S_\N}a_{\w}^{\sigma}.
\]
Then $(a_{\w}^{*} : \w \subseteq \N)$ satisfies the constraint  in \eqref{eq:LPBNS} and
\[
\sum_{\w \subseteq \N}a_{\w}^{*} = \sum_{\w \subseteq \N}a_{\w}.
\]
\end{corollary}

\begin{proof}
From  Proposition \ref{prop:2},  for any feasible solution $(a_\w : \w \subseteq \N)$ in \eqref{eq:LPBNS}, there exists $|S_\N|$ feasible solutions $(a_{\w}^{\sigma(i)} : \w \subseteq \N, \text{ and } \sigma^{(i)} \in S_\N)$. Since \eqref{eq:LPBNS} is a linear programming problem,  $a_{\w}^{*}$ (the average of all feasible solutions) is also a feasible solution.  
%Furthremore, as  $a_\w$ and $a_{\w}^{*}$ are ``permutations'' of each other,   $\sum_{\w \subseteq \N}a_\w$ and $\sum_{\w \subseteq \N}a_{\w}^{*}$  are equal. 
The result then follows.
\end{proof}

Proposition \ref{prop:2} can be used to reduce  the complexity in solving the optimisation problem  \eqref{eq:LPBNS}. 
%In particular, the proposition implies that it is sufficient to consider only feasible solutions $(a_{\w}^{*} : \w \subseteq \N)$.

\begin{proposition} \label{prop:3}
If $|\w| = |\s|$, then
$
a_{\w}^{*} = a_{\s}^{*}
$.
\end{proposition}
\begin{proof}
Direct verification due to symmetry.
\end{proof}

By Proposition \ref{prop:2}, it is sufficient to consider only ``symmetric'' feasible solution $(a_{\w}^{*} : \w \subseteq \N)$. Therefore, we can impose additional constraint 

\begin{align} \label{eq:additioncons}
\Aw = \As, \quad \forall |\w| = |\s|.
\end{align}
to \eqref{eq:LPBNS} without affecting the bound. 
%
%the maximum of \eqref{Eq:linprog1} will not be reduced even if we impose additional equality constraint
These equality constraint  will  significantly reduce the number of variables in the optimisation problem. Specifically, we have the following theorem.

\begin{theorem} \label{thm:SimpRLRC}
Consider a $(r,\beta,\Gamma,\zeta)$ robust locally repairable code $\C$. Then, $|\C|$ is upper bounded by the optimal value in the following maximisation problem
%\begin{equation} 
%\begin{array}{lll}
%\textbf{maximize}   \quad \sum_{t=0 }^{n} { n \choose t} a_{t} \\
%\textbf{subject to}  \\ 
%a_{t}  \ge 0 ,  \quad\forall  t =0, \ldots, n\\ 
%%
%b_{t} = \sum_{i=0}^{t}\sum_{j=0}^{n-t}{t \choose i}{n-t \choose j}a_{i+j} (-1)^{i}(q-1)^{t-i} \\
%\hspace{5.4cm} \forall  t =0, \ldots, n\\
%\displaystyle  
%b_{t} \ge 0, \quad\forall  t =0, \ldots, n\\
%%
%\sum_{t=1}^{\beta} a_{t}  = 0 &  \\
% a_{0} = 1   \\
%\displaystyle \sum_{t=1}^{r}  {n-1-\Gamma  \choose t} b_{t+1} \ge \zeta (q-1) \sum_{t=0}^{n} {n \choose t} a_{t}. 
%  \end{array} \label{Eq:linprog2}
%\end{equation}
%

\begin{align} 
&\textbf{Maximize}   \quad \sum_{t=0 }^{\n} { \n \choose t} a_{t}  \nonumber \\
&\textbf{Subject to}  \nonumber \\ 
&\begin{cases}
a_{t}  \ge 0 ,  &  \quad\forall  t =0, \ldots, \n \\ 
b_{t} = \sum_{i=0}^{t}\sum_{j=0}^{\n-t}{t \choose i}{\n-t \choose j}a_{i+j} (-1)^{i}(q-1)^{t-i} & \quad \forall  t =0, \ldots, \n\\
\displaystyle  
b_{t} \ge 0,  &\quad\forall  t =0, \ldots, \n\\
\sum_{t=1}^{\beta} a_{t}  = 0 &  \\
 a_{0} = 1  &  \\
\displaystyle \sum_{t=1}^{r}  {\n-1-\Gamma  \choose t} b_{t+1} \ge \zeta (q-1) \sum_{t=0}^{\n} {\n \choose t} a_{t}. 
  \end{cases} 
  \label{Eq:linprog2}
\end{align}
\end{theorem}
%---------------------- PROOF
\begin{proof}
Due to Proposition \ref{prop:3}, we can simplify the LP in \eqref{eq:LPBNS}  by  rewriting the variables as follows:

\begin{align*}
a_{t}  & = A_{\w} \quad \text{ for } t= |\w| \\
b_{t}  & = B_{\w}  \left( \sum_{\s \subseteq \N} A_{\s} \right) \quad \text{ for } t= |\w|.  
\end{align*}
Using the new variables,   the objective function (i.e., the number of codewords in code $\C$) is reduced to sum of the number of codewords with a support of size $t$ (i.e., $a_t$) multiplied by the number of subsets with size $t$ (i.e., $\n \choose t$) on all $t = 0, \ldots, \n$. 

The first constraint follows by the fact that the number of codewords with support of size $t$ is not negative. 
For any $\w \subseteq \N$ with the size of $t$, the second constraint corresponds to that the number of codewords in dual code $\Cd$ is nonnegative. 

To rewrite the second constraint, notice that 
\begin{align}\label{eqMW}
\prod_{j=1}^{\n} \kappa (s_j,w_j) = (-1)^{|\s \cap \w|} (q-1)^{|\w - \s|}
\end{align}
for any pair $(\s,\w)$ such that  $|\s \cap \w| = i$ and $|\w - \s| = t-i$.
In addition, the number of $\s$ such that $|\s \cap \w| = i$ and $|\w - \s| = t-i$ is equal to 
\[
{t \choose i}{\n-t \choose j}.
\]
%
%
%
%
%
%Any $\w$ of size $t$ has $t$ nonzero and $n-t$ zero positions. Then, we assume the size of any $\s \subseteq \N$ is $i+j$ where $i$ and $j$ are the nonzero positions of $\s$ a subset of $t$ and $n-t$, respectively. Therefore, the number of all subsets $\s$ for any subset $\w$ is $t \choose i$$n-t \choose j$ for all $i = 0, \cdots t$ and $j = 0, \cdots, n-t$. Then, the number of codewords in code $\C$ will be $a_{i+j}$. From equation \eqref{eq:prop1} 
%
The third constraint follows by the fact that the number of codewords in dual code $\Cd$ with support of size $t$ is not negative. The forth constraint follows from the GR criterion and is equivalent to the third constraint in \eqref{eq:LPBNS}. According to the GR criterion, there exist no codewords with the support size $\le \beta$ (i.e., $a_t = 0, ~1 \le t \le \beta$). Since $(a_t \ge 0, ~ \forall t = 0, \ldots, \n)$, then $\sum_{t=1}^{\beta} a_{t}  = 0$. The constraint $a_0 = 1$ is equivalent to the fourth constraint in \eqref{eq:LPBNS} which states that there exists only one codeword in code $\C$ with a support of size $t=0$ (i.e., zero codeword). The last constraint is equivalent to the last constraint in \eqref{eq:LPBNS} which follows from the RLR criterion where 
%. $b_{t+1}$ is the number of codewords in dual code (i.e., repair groups) with size $t+1$ which $t$ is the number of helper nodes. The 
there are $\n-1-\Gamma \choose t$ subsets (i.e., supports) of size $t$ chosen from  $(\n-1-\Gamma)$ survived nodes.
\end{proof}

\begin{remark}
The number of variables and constraint now scales only linearly with $\n$ (the number of nodes in the network).
\end{remark}

%%%%%%%%%%%%%%%%%%%%%%%%Section 5%%%%%%%%%%%%%%%%%%%%%%%%%%%%%%%%
\section{Cost for Update} \label{Sec:Update}

In previous sections,  criteria evaluating performance of storage codes include 1)  storage efficiency, 2) local recovery capability and 3) global recovery capability. However, data are rarely static but will be updated from time to time. The fundamental question is how efficient a scheme can be?

\subsection{Update Model and Code Definition}
To answer this question, we assume that the data $\M$ consists of multiple files of identical size. In our previous model, we assume that each file will be encoded using the same storage code. Therefore, to update each file, one needs make request to each individual storage node to update the content (regarding that updated file) stored there, leading to excessive overheads.

In this section, we consider a variation of our scheme in the sense that we treat each systematic codeword symbol as a single file (which can be updated individually). The whole collection of files will still be robust and efficient as in the original model. However, one can now consider the problem of update as a single systematic codeword symbol and evaluate how many parity-check codeword symbols are required to be updated.

\begin{definition}[Update-efficient storage code]
A $\k$-symbol update-efficient storage code is specified by a $\k\times (\n+\k)$ generator matrix ${\bf G}$ with columns $(\bg_{1} , \ldots, \bg_{\n+\k} )$ which satisfy the following criteria:
\begin{enumerate}
\item \label{item:permute}
\emph{Pseudo-Systematic --}
The submatrix formed by the columns  $(\bg_{\n+1} , \ldots, \bg_{\n+\k}) $  is
an identity matrix.

%The submatrix formed by the columns  $(\bg_{\n+1} , \ldots, \bg_{\n+\k}) $  is an $k\times k$ permutation matrix (an invertible matrix                          where all but one entries in each row  or column  are zero). For our purpose, we can usually assume without loss of generality that the matrix is the identity matrix.
\item \label{item:rank}
\emph{Full Rank:}
The matrix formed by $(\bg_{1}, \ldots, \bg_{\n})$ is full rank (i.e., of rank equal to $\k$). 
\end{enumerate}
\end{definition}

The first $\n$ columns correspond to the "coded symbols'' stored in the $\n$ storage nodes, while the last $\k$ columns correspond to the $\k$ source symbols. Let $[u_{1}, \ldots, u_{\k}]$ be  the $\k$  source symbols. In our context, each source symbol may correspond on a file which can be separately updated. 
Then, any codeword $\c=[c_1,\ldots,c_{\n},c_{\n+1},\ldots,c_{\n+\k}]$ in the update efficient storage code $\C$ is generated by
\[
c_i=\sum_{j=1}^{\k} u_{j} g_{j,i}, \qquad \forall i=1,\ldots,\n+\k,
\] 
where $\bg_{i} = [g_{1,i} , \ldots, g_{\k,i} ]^{\top}$. However, only the first $\n$ coordinates of the codeword $\c$ will be stored in the network. Therefore, for any $i \in \N$ the content of the storage node $X_i$, is the coded symbol $c_i$. In particular,

\begin{align}\label{eq:UEcode}
X_i=\sum_{j=1}^{\k} u_{j} g_{j,i}, \qquad \forall i=1,\ldots, \n.
\end{align}

For any vector $\v = [v_1,\ldots,v_{\n+\k}]$, we will often rewrite it as $\v=[\v_1,\v_2]$ where 

$$
\v_1=[v_1,\ldots,v_{\n}] \text{ and } \v_2=[v_{\n+1},\ldots,v_{\n+\k}].
$$

\begin{definition}[support]
For any vector $\v=[\v_1,\v_2]$ where 

$$
\v_1=[v_1,\ldots,v_{\n}] \text{ and } \v_2=[v_{\n+1},\ldots,v_{\n+\k}],
$$
its support  $\lambda(\v)$ is defined as  two subsets  $(\w_1 \subseteq \N, \w_2\subseteq \K)$ where 

\begin{align*}
i \in \w_1 &\text{ if }   v_{i} \neq 0  \\
i \in \w_2 &\text{ if }   v_{i+\n} \neq 0,
\end{align*}
and $\K=\{1,\ldots,\k\}.$
As before, it is also instrumental to view $\w_{1}$ and $\w_{2}$ as binary vectors 
$\w_1=[w_{1,1},\ldots,w_{1,N}]$ and $\w_2=[w_{2,1},\ldots,w_{2,\k}]$ such that 
$w_{1,i} = 1$ if $ v_{i} \neq 0$ and $w_{2,i} = 1$ if $ v_{i+\n} \neq 0$.
\end{definition}

\begin{figure}[t]
\centering
\includegraphics[width=.5\textwidth]{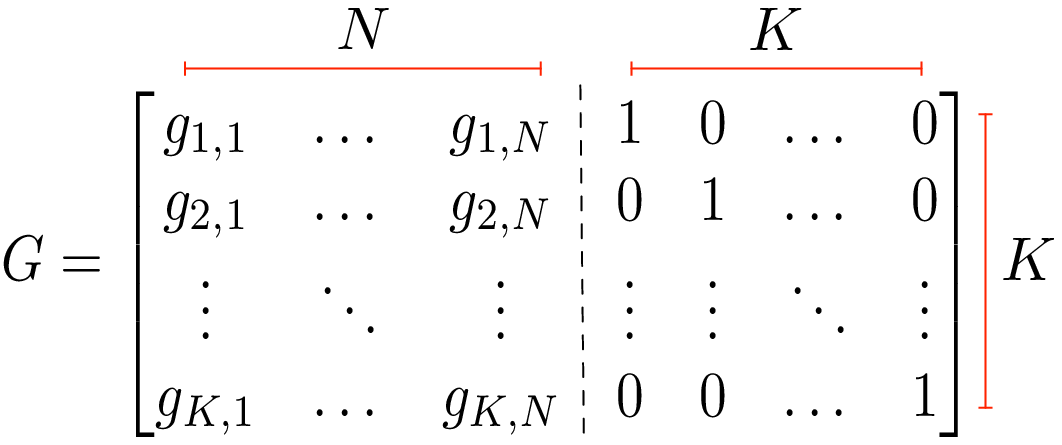}
\caption{Generator matrix of an update efficient storage code.}
\label{fig:UEG}
\end{figure}
\begin{figure}[t]
\centering
\includegraphics[width=.7\textwidth]{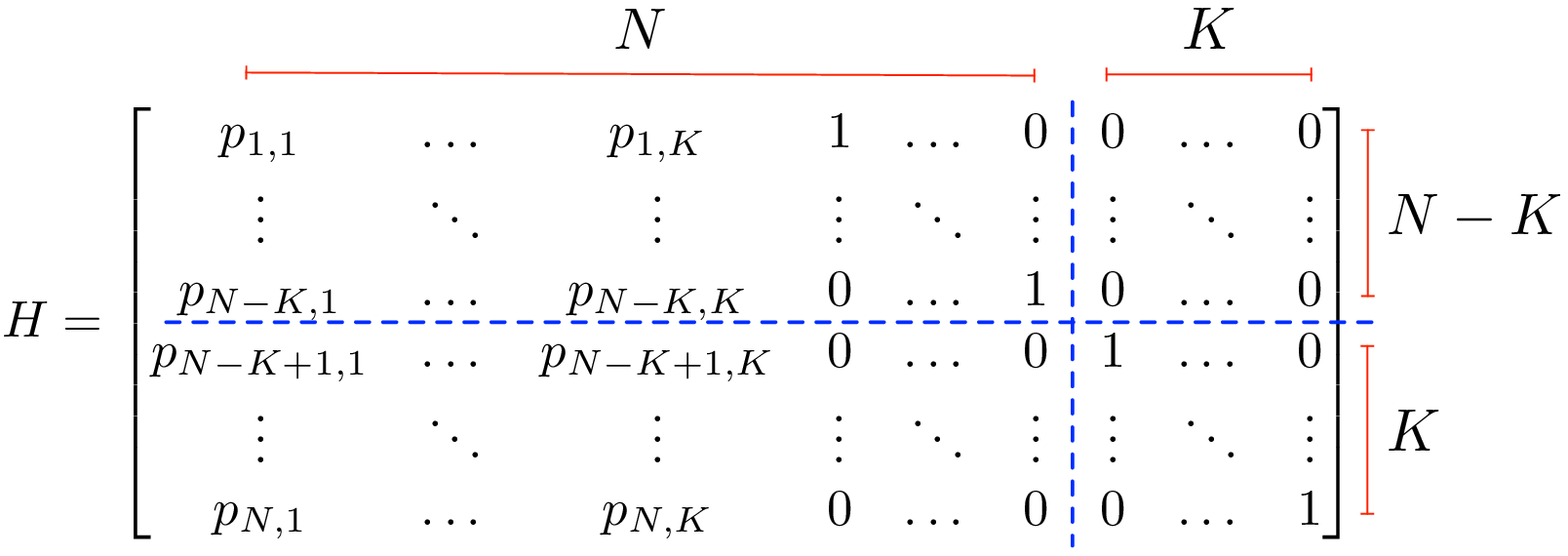}
\caption{Parity check matrix of an update efficient storage code.}
\label{fig:UEH}
\end{figure}

\begin{lemma}
Let $G$ be the generator matrix of an update efficient storage code, as depicted in Figure \ref{fig:UEG}. Then 
\begin{enumerate}

\item 
For any $i=1,\ldots, \k$, there exists a dual codeword $\h = (\h_{1}, \h_{2})$ such that all entries of $h_{2,j}$ is non-zero if and only if $j\neq i$. 

\item The rank of the last $\k$ column of the parity check matrix must be $\k$.

\item  If $\h = (\h_{1}, \h_{2})$ is a dual codeword, and $\h_{2}$ is non-zero, then $\h_{1}$ must be non-zero as well.

\item  If $\c = (\c_{1}, \c_{2})$ is a codeword, and $\c_{2}$ is non-zero, then $\c_{1}$ must be non-zero as well.

\end{enumerate}

\end{lemma}
%---------- PROOF
\begin{proof}
Rank of $G$ is clearly $\k$ and hence  rank of $H$ must be $\n$. By construction, the first $\n$ columns of $G$ has rank $\k$. This implies 1), which in turns also implies 2).  Next,  3) follows from that the last $\k$ columns of $G$ has rank $\k$. Finally, 4) follows from that the matrix formed by $(\bg_{1}, \ldots, \bg_{\n})$ is of rank $\k$.
\end{proof}

 \begin{remark}
Assume without loss of generality that the last $\n$ columns of the generator matrix $G$ has rank $\n$. We can assume that the parity check matrix $H$ of the update-efficient storage code can be  taken the form in Figure \ref{fig:UEH} such that for any $1 \le \ell \le \k$, the $\n-\k+\ell^{th}$ row of $H$ has at least two non-zero entries. In other words, entries $ p_{\n-\k+\ell, 1} ,  \ldots, p_{\n-\k+\ell, \k}$ cannot be all zero.   
 \end{remark}

\begin{proposition}[Code Properties]\label{prop4}
Let $\C$ be a $\k$-symbol update efficient locally repairable code. Then 

\medskip
\begin{align}\label{eqprop4}
\begin{cases}
(\ref{eqprop4}.a)  \quad  \sum_{\w_{2}\subseteq \K : |\w_{2}| \ge 1} \support_{\C}({\bf 0} , \w_2) = 0 \\ 
(\ref{eqprop4}.b)  \quad  \sum_{\w_{2}\subseteq \K : |\w_{2}| \ge 1} \support_{\C^{\perp}}({\bf 0} , \w_2) = 0     \\
(\ref{eqprop4}.c)  \quad \sum_{\w_{1}\subseteq \N  } \support_{\C^{\perp}} ( \w_{1},\w_2) = (q-1)^{|\w_{2}|}q^{\n-\k}, & \forall \w_2 \subseteq \K, |\w_2|\ge 0 \\
(\ref{eqprop4}.d)  \quad\sum_{\w_{1}\subseteq \N  } \support_{\C } ( \w_{1},\w_2) = (q-1)^{|\w_{2}|} , & \forall \w_2 \subseteq \K, |\w_2|\ge 0 \\
(\ref{eqprop4}.e)  \quad \support_{\C}({\bf 0} , {\bf 0} ) = \support_{\C^{\perp}}({\bf 0} , {\bf 0} ) = 1 %\text{ and } \support_{\C}(\w_{1} , {\bf 0} )  = 0, \: \forall \w_{1} \neq {\bf 0}
\end{cases}
\end{align}
\end{proposition}

\begin{proof}
Direct verification from the generator matrix and the parity check matrix in Figures \ref{fig:UEG} and \ref{fig:UEH}, respectively.
\end{proof}

%
%In this sense, the identity criterion ensures that \eqref{eq:UEcode} characterises the relation between the source symbols and the coded symbols, while  the full rank criterion guarantees that all source symbols can be recovered from the coded symbols stored in the $\n$ storage nodes. Again, notice that only the $\n$ coded symbols will be stored in the storage nodes.
%

\begin{definition}[Update efficient locally repairable storage code] \label{def:UELRC}
A $\k-symbol$ update efficient locally repairable code $\C$  is a $(r,\beta,\delta)$ storage code  if satisfies the following criteria:
\begin{enumerate}
\item \emph{Local Recovery (LR):}
for any $i\in \{1, \ldots, \n\}$, there exists ${\bf h}=[\h_1,\h_2] \in \C^{\perp}$ where $\h_1=[h_1,\ldots,h_{\n}]$ and $\h_2=[h_{\n+1},\ldots,h_{\n+\k}]$ such that $i\in \lambda({\bf h}_1)$, $\lambda(\h_2) =\emptyset$, and $|\lambda(\h_1)|-1 \le r$.  

\item \emph{Global Recovery (GR):} $\support_{\C}(\w_1,\w_2) = 0$, for all $\w_1 \subseteq \N$ such that $1 \le |\w_1| \le  \beta$ and $\w_2 \subseteq \K$.

\item \emph{Efficient Update (EU):} For any $\w_2 \subseteq \K$ and $|\w_2|=1$, 
\begin{align*}
\support_{\C}(\w_1,\w_2)= 0, \qquad \forall\w_1 \subseteq \N \text{~and~} |\w_1| > \delta.
\end{align*}
%\tc{
%and
%\[
%\sum_{\w_{1}\subseteq \N}\support_\C(\w_1,\w_2) = q-1, \qquad \forall \w_{1} \subseteq \N, |\w_{1}| > \delta. 
%\]
%}
%
\end{enumerate} 
\end{definition}

\begin{remark}
It is worth to notice that these parameters may depend on each other. For example, it can be directly verified that  there is no $(r,\beta,\delta)$ storage code where $\beta \ge \delta$.

%Since only the first $\n$ symbols of each codeword in $\C$ are stored in the network, then the 
The LR criterion guarantees that for any failed node $X_i$ there exists a local repair group specified by the support of the first $\n$ symbols of the corresponding dual codeword. In other words, while
\[
\h\c^\top=\sum_{i=1}^{\n+\k}h_ic_i=0,
\]  
only $(h_1,\ldots,h_{\n})$ and $(c_1,\ldots,c_{\n})$ are participating in the repair process such that
\[
\sum_{i=1}^{\n}h_iX_i=0.
\]  
Therefore, the support of $\h_1=[h_1,\ldots,h_{\n}]$ specifies the repair groups while $\h_2$ must be a zero vector. The GR criterion guarantees that the stored data is recoverable at the presence of severe failure pattern. More specifically,  in the worst case of  $\beta$ simultaneous node failures, data stored in the network can still be repaired from the survived nodes, i.e., minimum distance of the code is $d=\beta+1$.
%Equivalently, it means that there are no codewords with the first $\n$ symbols support $\w_1$ of size equal or smaller than $\beta$. In particular, the minimum distance of the code is $d_{\textnormal{min}}=\beta + 1$. 
Finally, the EU criterion ensures that if any of the source symbols is modified, then no more than $\delta$ storage nodes need to be updated. In particular, for any singleton $\w_{2}$, there exists a unique $\w_{1}$ such that 
\[
\support_{\C}(\w_1,\w_2) = \sum_{\w'_{1}}\support_{\C}(\w'_1,\w_2) = q-1.
\] 
The $\w_{1}$ corresponds to  codewords where only the single data file (indexed by the support of $\w_{2}$) is non-zero.  
In that case, if the data file has been changed, only the codeword symbols corresponding to the support of $\w_{1}$ need to be updated.
% 
%
%Any codeword with support $\w_2$ for $|\w_2|=1$ indicates the relation of the corresponding source symbol indexed in $\w_2$ with the encoded symbols (i.e., storage nodes content) indexed in $\w_1$. 
%
% Therefore, if the corresponding symbol is changed, then all relative nodes must be updated. 
% 
 In order to maintain the update cost efficiency (i.e., equal or less than $\delta$ nodes per any source symbol change), the code $\C$ must not include any codeword with support $|\w_1|>\delta$ for any $|\w_2|=1$.

%%%%%%%%%%%%%%%%%%%%%%%%%%%%%%
\begin{proposition}\label{prop5}
For a $(r,\beta,\delta)$ update efficient locally repairable code $\C$, we have 

\medskip
\begin{align}\label{propeq5}
\begin{cases}
(\ref{propeq5}.a) \quad 
\sum_{\w_{1}\subseteq \N : \ell \in \w_{1} \text{ and }|\w_{1}| \le r+1} \support_{\C^{\perp}} ( \w_{1},\emptyset) \ge (q-1),
\qquad \qquad \qquad \forall \ell \in\N  \\ % \label{eq:NecConAS4} \\
(\ref{propeq5}.b) \quad 
\support_{\C} (\w_{1},\w_{2})   =0,  
\qquad \qquad \qquad \qquad \forall \w_1 \subseteq \N, \w_2 \subseteq \Kp, 1 \le |\w_{1}| \le \beta  \\% 
(\ref{propeq5}.c) \quad 
\support_{\C}(\w_1,\w_2)= 0, 
\qquad \qquad \qquad \qquad \forall\w_1 \subseteq \N \text{~and~} |\w_1| > \delta \text{~and~} |\w_{2}|=1
%
%\sum_{\w_{1}\subseteq \N}\support_\C(\w_1,\w_2) = q-1, 
%& \forall \w_{1} \subseteq \N, |\w_{1}| > \delta. 
\end{cases}
\end{align}
\end{proposition}

%
%Now we define our update efficient storage code. We still assume that our update efficient code has the local repairability property. Notice that in the following definition we do not consider the robustness of the locally repairable codes for computational simplicity. However, Definition \ref{def:UELRC} can easily be generalised to the update efficient robust locally repairable codes by using Definition \ref{def:RLRC}.
\end{remark}

\begin{remark}
For notation convenience, we will first consider only the case of local recovery.  However, the extension to include robust local recovery is straightforward. We will present the results afterwards.
\end{remark}

The following theorem gives the necessary conditions for existence of an update efficient locally repairable code based on its generator matrix specifications in Definition \ref{def:UELRC}.

%criteria \ref{item:permute} and \ref{item:rank} in  

\begin{theorem}[Necessary condition]\label{thm:4}
Let $\C$ be a  $(\n,\k, r,\beta,\delta)$ update-efficient locally repairable storage code. Then there exists
$(A_{\w_{1},\w_{2}} , C_{\w_{1},\w_{2}} : \: \w_{1 }\subseteq \N, \w_{2} \subseteq \Kp)$ such that  

\begin{align} \label{eqC1to12}
\begin{cases}
%%%%%%%%% MacWilliams
(\text{C1}) \quad  
C_{\w_{1},\w_{2}}   =  {\sum_{\bs_{1} \subseteq \N, \bs_{2} \subseteq \Kp } A_{\bs_{1},\bs_{2}}  \prod_{j=1}^{\n+\k}\kappa_{q}(s_j,w_j)}, 
& \quad \forall \w_{1} \subseteq \N, \w_{2} \subseteq \Kp  \\
%%%%%%%%%%% Non-negativity
(\text{C2}) \quad 
A_{\w_1,\w_2}  \geq 0, 
& \quad \forall \w_1 \subseteq \N,\w_2 \subseteq \K  \\ %\label{eq:NonnegCode} \\
(\text{C3}) \quad 
C_{\w_1,\w_2}  \geq 0,  
& \quad \forall \w_1 \subseteq \N,\w_2 \subseteq \K \\ %\label{eq:NonnegDualCode} \\
%%%%%%%%%%%%%%
%%%%%%%        Code size
%%%%%%%%%%%%%%
(\text{C4}) \quad    A_{\emptyset,\w_2} = 0  & \quad \forall \w_{2} \subseteq \K, |\w_{2}| \ge 1 \\ 
(\text{C5}) \quad    C_{\emptyset,\w_2} = 0  & \quad \forall \w_{2} \subseteq \K, |\w_{2}| \ge 1 \\
(\text{C6}) \quad   \sum_{\w_{1}\subseteq \N  } C_{\w_1,\w_2}  = (q-1)^{|\w_{2}|}q^{\n },  & \quad \forall \w_2 \subseteq \K, |\w_2|\ge 0 \\
(\text{C7}) \quad   \sum_{\w_{1}\subseteq \N  } A_{\w_1,\w_2} = (q-1)^{|\w_{2}|} ,  & \quad \forall \w_2 \subseteq \K, |\w_2|\ge 0 \\
%%%%%%%% Zero codewords
(\text{C8}) \quad    A_{\emptyset,\emptyset}  = 1, &  \\ % \label{eq:NecConAS7}
%%%%%%%%%%%%%%%%
(\text{C9}) \quad   C_{\emptyset,\emptyset}  =q^{\k}, & \\
%%%%%%%%%%%%%%%%%%%%%%%%%
%
%
% The following constraints are due to the reliability
%%%%%%%%%%%%%%%%%%%%%%%%
%%%%%%%% LR
(\text{C10}) \quad  
\sum_{\w_{1}\subseteq \N : \ell \in \w_{1} \text{ and }|\w_{1}| \le r+1} C_{ \w_{1},\emptyset} \ge (q-1) q^{\k}   ,
& \quad \forall \ell \in\N  \\ % \label{eq:NecConAS4} \\
%%%%%%%%%%%% GR
(\text{C11}) \quad  
A_{\w_{1},\w_{2}}   =0, 
& \quad \forall  1 \le |\w_{1}| \le \beta  \\% \label{eq:NecConAS5} \\
%%%%%%%% EU
(\text{C12}) \quad 
 A_{\w_{1},\w_2}  = 0 , 
& \quad \forall   |\w_{1}| > \delta, |\w_2|=1 \\%\label{eq:NecConAS6} \\
%A_{\w_1,\emptyset}  =0,  & \quad \forall\w_1\subseteq \N, \w_1\neq \emptyset \\
%%%%%%%%%%%%%%%%
%
\end{cases}
\end{align}
\end{theorem}
%------------- PROOF 
\begin{proof}%[Proof of Theorem \ref{thm:4}]
Let $\C$ be a $(r,\beta,\delta)$ update efficient linear storage code. We define

\begin{align}
A_{\w_{1},\w_{2}} & \triangleq \support_{\C}(\w_{1},\w_{2}) \label{eq23}\\
C_{\w_{1},\w_{2}} & \triangleq q^{\k} \support_{\C^{\perp}}(\w_{1},\w_{2}) .\label{eq24}
\end{align}
Constraint (C1)  follows directly from the MacWilliams identity in Proposition \ref{prop:duality}. Constraint (C2)  and (C3) follow from \eqref{eq23} and \eqref{eq24} and that enumerating functions are non-negative. 
%Notice that here, $\w_1 \subseteq \N=\{1,\ldots,N\}$, $\w_2 \subseteq \K=\{N+1,\ldots,\n+\k\}$, and $\w=\w_1\cup \w_2$. 
Constraints (C4) - (C9) are consequences of Proposition \ref{prop4} and 
Constraints (C10)-(C12) are due to Proposition \ref{prop5}.
\end{proof}
%
%follow from the fact that enumerator functions are non-negative. According to the definition of $\C$, the size of the code and its dual is respectively $q^{\k}$ and $q^N$ which are identified by constraint \eqref{eq:NecConCodeSize} and \eqref{eq:NecConDualCodeSize}. Constraints \eqref{eq:NecConAS2}-\eqref{eq:NecConAS6} can be directly verified from criteria \eqref{eq:PermuteThm}-\eqref{eq:EUthm} in Theorem \ref{thm:3}, respectively. The constraints \eqref{eq:NecConAS7}-\eqref{eq:NecConASN4} respectively follow from the fact that there exists only one zero-codeword in $\C$ and $\Cd$. According to the generator matrix $G$ of the code $\C$ in Figure \ref{fig:UEG}, the support $\w_2$ of any codeword $\c \in \C$ determines the corresponding rows of $G$ such that their linear combination is $\c$. Therefore, there exist $(q-1)^{|\w_2|}$ codewords $\c$ with the same support $\w_2$. Then, \eqref{eq:NecConASN1} follows. 

As before, the constraints in \eqref{eqC1to12}  are overly complex. Invoking symmetry, we will to simplify the above set of constraints.  
%
% to reduce the exponential complexity of the constraints in Theorem \ref{thm:4} to linear, we apply symmetry to the problem. The procedure is fairly similar to one in Section \ref{sec:SymCompx}. 

\begin{proposition} \label{prop:PermuteCodeUE}
Suppose the generator matrix specified with the column vectors $(\bg_{1} , \ldots, \bg_{\n+\k} )$ defines a $(r,\beta,\delta)$ code $\C$. For any $\sigma_{1} \in S_{\N}$ (i.e., symmetric group of $\N$) and $\sigma_{2} \in S_{\Kp}$ (i.e., symmetric group of $\Kp$), let $\C^{\sigma_{1},\sigma_{2}}$ be another code specified by 
$
(\bff_{i}, i=1, \ldots, \n+\k)
$
such that

\begin{align*}
\bff_{i} & =  \bg_{\sigma_{1}(i)}  \quad  \text{ for  } i \in \N \\ 
\bff_{\n+j } & =  \bg_{\n+\sigma_{2}(j)} \quad   \text{ for } j \in \K. 
\end{align*}
Then  $\C^{\sigma_{1},\sigma_{2}}$ is still a  $(r,\beta,\delta)$ code.
\end{proposition} 

\begin{proof}
Let $\c=[\c_1,\c_2]$ be a codeword in $\C$ where $\c_1=[c_1,\ldots,c_{\n}]$ and $\c_2=[c_{\n+1},\ldots,c_{\n+\k}]$.  Then, 
%using $\bff_i$ for $i=1,\ldots,\n+\k$  as the column vectors of the generator matrix of $\C^{\sigma_1,\sigma_2}$, 
the vector $\c^{\sigma_1,\sigma_2}=[\c_1^{\sigma_1},\c_2^{\sigma_2}]$ is a codeword in $\C^{\sigma_1,\sigma_2}$ where 
$\c_1^{\sigma_1}=[\sigma_1(c_1),\ldots,\sigma_1(c_{\n})]$ and 
$\c_2^{\sigma_2}=[c_{\n+\sigma_2({1})},\ldots, c_{\n+\sigma_2({\k})}]$. 
In other words, any codeword $\c^{\sigma_1,\sigma_2} \in \C^{\sigma_1,\sigma_2}$ is a permutation of a codeword $\c \in \C$ for all $\sigma_1 \in S_\N$ and $\sigma_2 \in S_\K$. 
Equivalently, the new code $\C^{\sigma_1,\sigma_2}$ is obtained from $\C$ by relabelling or reordering the codeword indices. The proposition thus follows.    
\end{proof}

\begin{corollary}\label{cor:1}
Suppose  $(a_{\w_{1},\w_{2}} , c_{\w_{1},\w_{2}} :\: \w_{1 }\subseteq \N, \w_{2} \subseteq \Kp)$ satisfies \eqref{eqC1to12}. Let 

\begin{align}
a^{*}_{\w_{1},\w_{2}} & = \frac{1}{|S_{\N}||S_{\Kp}|}\sum_{\sigma_{1} \in S_{\N}}\sum_{\sigma_{2} \in S_{\Kp}} a_{\sigma_{1}(\w_{1}), \sigma_{2}(\w_{2})  } \label{eq:AvgUE1} \\
c^{*}_{\w_{1},\w_{2}} & = \frac{1}{|S_{\N}||S_{\Kp}|}\sum_{\sigma_{1} \in S_{\N}}\sum_{\sigma_{2} \in S_{\Kp}} c_{\sigma_{1}(\w_{1}), \sigma_{2}(\w_{2})}. \label{eq:AvgUE2}
\end{align} 
Then $(a^{*}_{\w_{1},\w_{2}} , c^{*}_{\w_{1},\w_{2}} :\: \w_{1 }\subseteq \N, \w_{2} \subseteq \Kp)$  also satisfies \eqref{eqC1to12}. 
Furthermore, for any $\w_{1}, \bs_{1} \subseteq \N$ and $\w_{2}, \bs_{2} \subseteq \Kp$ such that 
$|\w_{1}| = |\bs_{1}|$ and $|\w_{2}| = |\bs_{2}|$. Then 

\begin{align}
a^{*}_{\w_{1},\w_{2}} & = a^{*}_{\bs_{1},\bs_{2}}\\
c^{*}_{\w_{1},\w_{2}} & = c^{*}_{\bs_{1},\bs_{2}}.
\end{align}
\end{corollary}

\begin{proof}
From Proposition \ref{prop:PermuteCodeUE}, if $(a_{\w_{1},\w_{2}} , c_{\w_{1},\w_{2}} :\: \w_{1 }\subseteq \N, \w_{2} \subseteq \Kp)$ satisfy constraints \eqref{eqC1to12}, then $(a_{\sigma_1(\w_{1}),\sigma_2(\w_{2})} , c_{\sigma_1(\w_{1}),\sigma_2(\w_{2})} :\: \w_{1 }\subseteq \N, \w_{2} \subseteq \Kp)$ will satisfy them for all $\sigma_1 \in S_\N$ and $\sigma_2 \in \K$. Since all the constraints in \eqref{eqC1to12} are convex, any convex combination of all these feasible solutions is still feasible. Hence, 
 $(a^{*}_{\w_{1},\w_{2}} , c^{*}_{\w_{1},\w_{2}} :\: \w_{1 }\subseteq \N, \w_{2} \subseteq \Kp)$ is still feasible. 
Also, 
\[
\left\{(\sigma_1(\w_1),\sigma_2(\w_2)): \forall \sigma_1 \in S_\N, \sigma_2 \in S_\K\right\}=\left\{(\sigma_1(\s_1),\sigma_2(\s_2)): \forall \sigma_1 \in S_\N, \sigma_2 \in S_\K\right\}
\]
if $|\w_{1}| = |\bs_{1}|$ and $|\w_{2}| = |\bs_{2}|$ for any $\w_{1}, \bs_{1} \subseteq \N$ and $\w_{2}, \bs_{2} \subseteq \Kp$.
\end{proof}

%In particular, Corollary  \ref{cor:1} implies that if $(a_{\w_{1},\w_{2}} , c_{\w_{1},\w_{2}} :\: \w_{1 }\subseteq \N, \w_{2} \subseteq \Kp)$ satisfy the constraint then $(a_{t_{1},t_{2}} , c_{t_{1},t_{2}} :\: t_1 \in \{0,\ldots,N\}, t_2 \in \{0,\ldots,\k\})$, where $a_{t_1,t_2}$ and $c_{t_1,t_2}$ can be interpreted as the average number of codewords and dual codewords with support of size $(t_1,t_2)$, respectively.  

\begin{theorem}[Necessary condition]\label{thm:5}
Let $\C$ be a  $(\n,\k, r,\beta,\delta)$ update-efficient storage code. Then there exists  real numbers $(a_{t_{1},t_{2}} , c_{t_{1},t_{2}} :\: t_{1 } = 0, \ldots,\n, t_{2} = 0, \ldots, \k )$ satisfying \eqref{eqD1to12}. 
%
%\begin{align}
%%%%%%%% MacWilliam
%\tc{(C1)} \quad 
%c_{t_{1},t_{2}} & = \sum_{u_{1}=0}^{t_{1}}\sum_{v_{1}=0}^{N-t_{1}}\sum_{u_{2}=0}^{t_{2}}\sum_{v_{2}=0}^{k-t_{2}} {t_{1} \choose u_{1} } {N-t_{1} \choose v_{1}} {t_{2} \choose u_{2} } {k-t_{2} \choose v_{2}}a_{u_1+v_1,u_2+v_2} \nonumber \\ 
%& (-1)^{u_{1} + u_{2}}(q-1)^{t_{1}+t_{2}-u_{1}-u_{2}}, \quad \forall 0\le t_{1} \le \n, 0\le t_{2} \le k \label{eq:MacwilliamUELPsym} \\
%\end{align}

\begin{align}\label{eqD1to12} 
\begin{cases}
(D1) \quad
c_{t_{1},t_{2}}  = \sum_{u_{1}=0}^{t_{1}}\sum_{v_{1}=0}^{\n-t_{1}}\sum_{u_{2}=0}^{t_{2}}\sum_{v_{2}=0}^{\k -t_{2}}  \Xi(u_{1},v_{1},u_{2},v_{2}) a_{u_1+v_1,u_2+v_2} \\
%%%%%%%Non-negativity
(D2) \quad 
a_{t_1,t_2}  \geq 0,   \qquad \qquad \qquad \qquad \qquad   \forall 0\le t_{1} \le \n, 0\le t_{2} \le \k   \\
(D3) \quad 
c_{t_1,t_2}  \geq 0,  \qquad \qquad \qquad \qquad \qquad \forall 0\le t_{1} \le \n, 0\le t_{2} \le \k \\
(D4) \quad 
%\sum_{t_2=1}^{k} 
a_{0,t_2}  = 0, \qquad \qquad \qquad \qquad \qquad \forall t_{2} \ge 1 \\
(D5) \quad
%\sum_{t_2=1}^{k} 
c_{0,t_2}  = 0,  \qquad \qquad \qquad \qquad \qquad \forall t_{2} \ge 1 \\
(D6) \quad 
\sum_{t_{1}=0}^{\n} {\n \choose t_{1}} c_{t_{1},t_{2}}  = (q-1)^{t_{2}}q^{\n} ,    
\qquad \qquad \qquad \qquad \qquad \forall t_{2}  \\
(D7) \quad 
\sum_{t_{1}=0}^{\n} {\n \choose t_{1}} a_{t_1,t_{2}}   = (q-1)^{t_{2}}, 
\qquad \qquad \qquad \qquad \qquad \forall t_{2}  \\
(D8) \quad
a_{0,0}  = 1,   \qquad \qquad \qquad \qquad \qquad \\
(D9) \quad
c_{0,0} = q^{\k}   \qquad \qquad \qquad \qquad \qquad \\
%%%%%LR
(D10) \quad
\sum_{t_{1}=2 }^{r+1}{\n-1 \choose t_{1}-1}  c_{ t_{1},0} \ge (q-1)q^{\k} \\
%\sum_{\w_{1}\subseteq \N : \ell \in \w_{1} \text{ and }|\w_{1}| \le r+1} C_{ \w_{1},\emptyset} \ge (q-1) q^{\k}  \\
%
%(q-1)q^{\k} \leq \sum_{t_{1}=2 }^{r+1}{N-1 \choose t_{1}-1}  c_{ t_{1},0} \leq q^N ,  & \\
%%%%%GR
(D11) \quad  
a_{t_{1},t_{2}}  = 0,   \qquad \qquad \qquad \qquad \qquad  \forall 
1\le t_{1} \le \beta \\
%%%%%EU
(D12) \quad  
a_{t_1,1}  = 0 , 
\qquad \qquad \qquad \qquad \qquad  \forall   t_{1} > \delta 
\end{cases}
\end{align}
where 
\[
\Xi(u_{1},v_{1},u_{2},v_{2})  = 
(-1)^{u_{1} + u_{2}}(q-1)^{t_{1}+t_{2}-u_{1}-u_{2}}
{t_{1} \choose u_{1} } {N-t_{1} \choose v_{1}} {t_{2} \choose u_{2} } {k-t_{2} \choose v_{2}}.
\]
\end{theorem}
%-------------- Proof
\begin{proof}
The theorem is   essentially a  direct consequence of Theorem \ref{thm:4} and Corollary \ref{cor:1}, which guarantees that there exists a ``symmetric'' feasible solution 
$(a^{*}_{\w_{1},\w_{2}} , c^{*}_{\w_{1},\w_{2}} :\: \w_{1 }\subseteq \N, \w_{2} \subseteq \Kp)$ satisfying \eqref{eqC1to12}. Due to symmetry, 
we can rewrite 
$a^{*}_{\w_{1},\w_{2}}$ as  $a_{t_{1}, t_{2}}$
and 
$c^{*}_{\w_{1},\w_{2}}$ as $c_{t_{1}, t_{2}}$
where $t_{1} = |\w_{1}|$ and $ t_{2} = |\w_{2}|$. 
The set of constraints in \eqref{eqD1to12}  is basically obtained by rewriting 
\eqref{eqC1to12} and grouping like terms.  
Most rephrasing is straightforward. The more complicated one is (D1), where we will show how to rewrite it. 
 
Recall the constraint (C1) in \eqref{eqC1to12}

\begin{align}\label{eq31}
C_{\w_{1},\w_{2}}   =  {\sum_{\bs_{1} \subseteq \N, \bs_{2} \subseteq \Kp } A_{\bs_{1},\bs_{2}}  \prod_{j=1}^{\n+\k}\kappa_{q}(s_j,w_j)}, 
 \quad \forall \w_{1} \subseteq \N, \w_{2} \subseteq \Kp  .
\end{align}
%Using \eqref{eq:prop1} in Proposition \ref{prop:duality}, we have 
%\begin{align}
%C_{\w_1,\w_2} = \sum_{\s_1 \subseteq \N, \s_2 \subseteq \K}A_{\s_1,\s_2} (-1)^{|(\s_1\cup\s_2) \cap (\w_1\cup\w_2)|}(q-1)^{|((\N \cup \K)\setminus(\s_1\cup \s_2)) \cap (\w_1 \cup \w_2)|}.
%\end{align}
Let $|\w_1|=t_1$, $|\w_2|=t_2$. For any  

 \begin{align}
 0\le u_{1} \le t_{1}, \: \text{ and } 0\le v_{1} \le \n-t_{1} \\
 0\le u_{2} \le t_{2}, \: \text{ and } 0\le v_{2} \le \k-t_{2},
 \end{align} 
%  $0\le u_{1} \le t_{1}, \: 0\le v_{1} \le \n-t_{1}, 0\le u_{1} \le t_{1}, \: 0\le v_{1} \le \n-t_{1},$
the number of pairs of  $(\s_{1},\s_{2})$  such that    
  
  \begin{align*}
 u_{1} = | \w_{1} \cap \s_{1}| \text{ and }  v_{1} = | \s_{1} \setminus \w_{1} |    \\ 
 u_{2} = | \w_{2} \cap \s_{2}| \text{ and }  v_{2} = | \s_{2} \setminus \w_{2} | 
 \end{align*}
 is equal to 
 \[
 {t_{1} \choose u_{1} } {\n -t_{1} \choose v_{1}} {t_{2} \choose u_{2} } {\k-t_{2} \choose v_{2}}.
 \]
Also, for each such pair,  by \eqref{eqMW}, 
\[
 \prod_{j=1}^{\n+\k}  \kappa_{q}(s_j,w_j) = (-1)^{u_{1} + u_{2}}(q-1)^{t_{1}+t_{2}-u_{1}-u_{2}}.
\]  
Grouping all like terms together, \eqref{eq31} can be written as in (D1) in \eqref{eqD1to12}.
\end{proof}

It is worth mentioning that unlike the linear programming problem in Theorem \ref{thm:SimpRLRC}, the code dimension $K$ (hence, the code size $|\C|=q^K$) is predetermined in Theorem \ref{thm:5}.
Notice that the complexity of the constraints in Theorem \ref{thm:4} is increased exponentially by the number of the storage nodes $\n$ and the number of the information symbols $\k$. However, the symmetrisation technique which is applied to the constraints significantly reduces the number of the variables to $2(\n+1)(\k+1)$. 

\begin{remark} 
Theorem \ref{thm:5} proves that if there exist no real numbers $(a_{t_1,t_2}, c_{t_1,t_2}~:~0\leq t_1\leq N, 0\leq t_2 \leq K)$ satisfying constraints (D1)--(D12) in \eqref{eqD1to12}, then there exist no $(N,K,r,\beta,\delta)$ update-efficient storage code.
\end{remark}

%Regarding the structure of the generator and parity matrices $G$ and $H$ of the update efficient code, respectively, in Figures \ref{fig:UEG} and \ref{fig:UEH}, all of the constraints in Theorem \ref{thm:4} (excluding the update constraint in \eqref{eq:NecConAS6}) are compatible with the repair locality of the code. Consequently, the constraint \eqref{eq:NecConAS6} guarantees the update efficiency of the locally repairable codes. 

\begin{remark}
In the previous linear programming bound, we consider only the local recovery criteria. 
Extension to robust local recovery criteria is similar.  

For the robust local recovery criteria, recall from Definition \ref{def:RLRC}, it requires that for any failed node $i\in\N$ and additional $\Gamma$ failed nodes $\gamma$, there exists at least $\zeta$ dual codewords (with different supports).
%for any $i \in \N$ and $\gamma \subset \N \setminus i$ such that $|\gamma| = \Gamma$, there exists $\h_{1}, \ldots , \h_{\zeta} \in \Cd$ such that for all $j=1, \ldots, \zeta$, 
%\begin{enumerate}
%\item 
%$i \in \lambda(\h_{j})$, $\gamma \cap \lambda(\h_{j}) = \emptyset$, and $|\lambda(\h_{j})|-1 \le r$.
%\item 
%$\lambda(\h_{j}) \neq \lambda(\h_{\k})$ for $k \neq j$.
%\end{enumerate}
%In our context, it means that for any failed node $i\in\N$ and additional $\Gamma$ failed nodes $\gamma$, there exists at least $\zeta$ dual codewords (with different supports) such that 

In other words, for any $i\in \N$ and $\gamma \subset \N \setminus i$ such that $|\gamma| = \Gamma$,

\begin{align}
\sum_{\w_{1} \subseteq \N,  i\in\w_{1}, \gamma \cap \w_{1} =\emptyset} 
\support_{\C^{\perp}} (\w_{1} , \emptyset)  \ge \zeta(q-1).
\end{align}
As before, the constraint can be rewritten as follows, after symmetrisation

\begin{align}
\sum_{t_{1} = 2}^{ r+1} {\n - \Gamma- 1 \choose t_{1} - 1} c_{t_{1},0} \ge \zeta (q-1) q^{\k}.
\end{align}

\end{remark}

%%%%%%%%%%%%%%%%%%%%%%%%Section 6%%%%%%%%%%%%%%%%%%%%%%%%%%%%%%
\section{Numerical Results And Discussion} \label{Sec:results}

In this section, we present two examples of  robust locally repairable codes. We will show that  these codes can repair a failed node locally when some of the survived nodes are not available. Using our bound, we prove that these codes are optimal.

\begin{example} \label{exam:1}
Consider a binary linear storage code of length $n=16$ with $k=9$ as defined by the following parity check equations

\begin{align} \label{Eq:P1}
0 & = P_1  +C_1 + C_2 + C_3 \\ 
0 & = P_2 + C_4 + C_5 + C_6 \\
0 & = P_3+ C_7 + C_8 + C_9 \\ \label{Eq:P4}
0 & = P_4+C_1 + C_4 + C_7 \\
0 & = P_5+ C_2 + C_5 + C_8 \\
0 & = P_6+ C_3 + C_6 + C_9 \\ \label{Eq:P7} \nonumber
0 & = P_7+ C_1 + C_2 + C_3 \\ \nonumber
& \qquad + C_4 + C_5 + C_6 \\ 
& \qquad + C_7 + C_8 + C_9.
\end{align}
Here,  $(C_{1}, \ldots, C_{9}, P_{1}, \ldots, P_{7})$ correspond to the content stored at the 16 nodes. 
In particular, we might interpret that $C_{1} , \ldots, C_{9}$ are the systematic bits while $P_{1}, \ldots, P_{7}$ are the parity check bits. The code can also be represented by Figure \ref{fig:Chan_Code}.
In this case, the parity check equations are equivalent to that the sum of rows  and the sum of columns are all zero. 

% is a binary linear storage code of length $n=16$ with $k=9$ data blocks  and   7 parity blocks. 
%We aim to construct a code such that a single failed node can be repaired locally using repair groups with size $r=3$ nodes if one extra node  is not available during the repair process. 
According to  \eqref{Eq:P1}-\eqref{Eq:P7}, every failed node (either systematic or parity) can be repaired by two different repair groups of size $r=3$. For instance, if $C_1$ is failed, it can be repaired by repair group $R_{C_1}^1=\{C_2, C_3, P_1\}$   or repair group $R_{C_1}^2=\{C_4, C_7, P_4\}$. Notice that the two repair groups are disjoint. Therefore, even if one of the repair groups is not available, there exists an alternative repair group to locally repair the failed node. 

\begin{figure}[t]
\centering
\includegraphics[width=.35\textwidth]{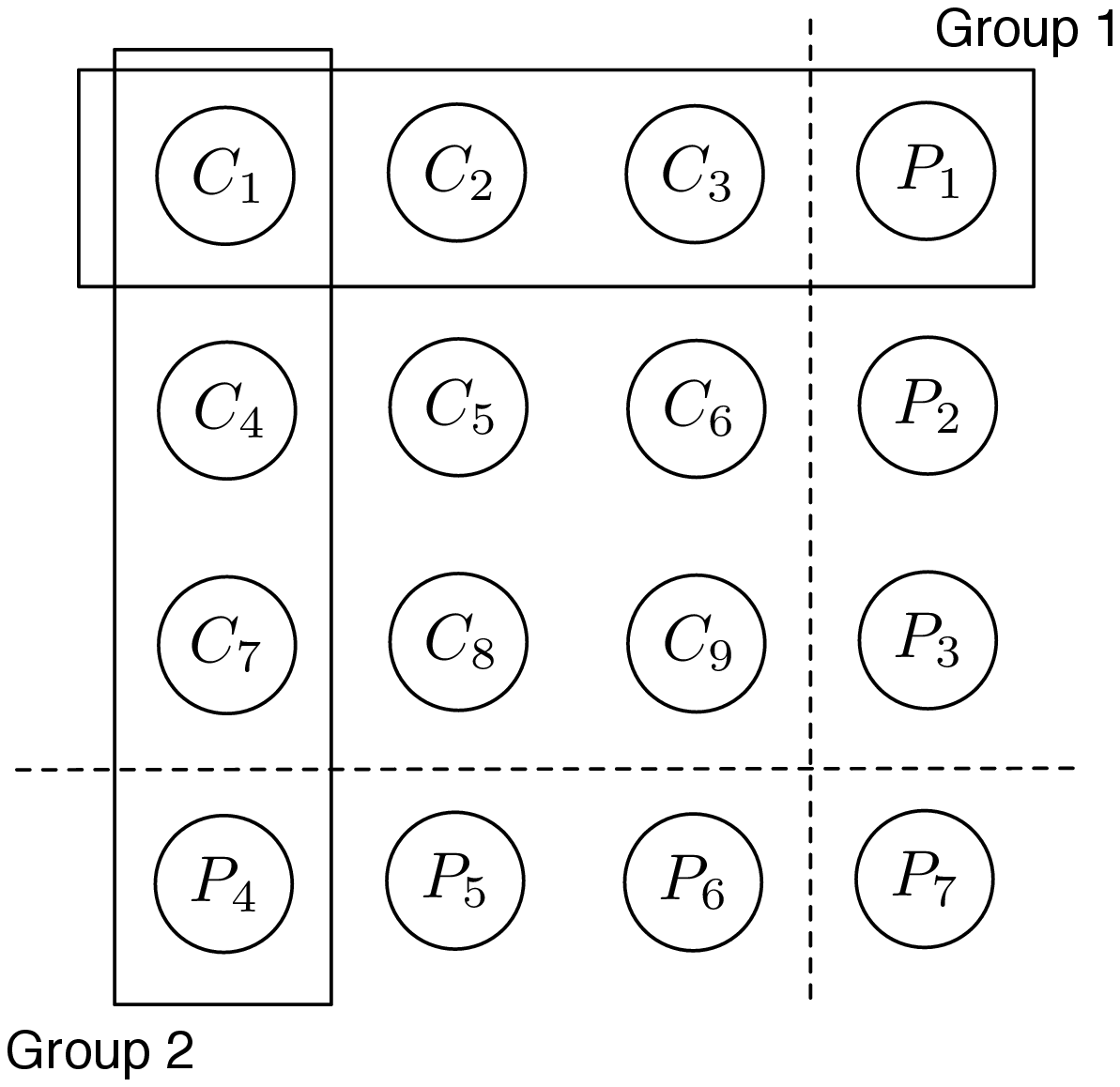}
\caption{A (3,3,1,1) robust locally repairable code of length $\n=16$ and dimension $\k=9$.}
\label{fig:Chan_Code}
\end{figure}

It can be verified easily that our code has a minimum distance of 4 and it is a $(3,3,1,1)$ robust locally repairable code. Therefore, any failed node can be repaired by at least $\zeta=1$ repairing group of size $r=3$ in the presence of any $\Gamma=1$ extra failure. Also, the original data can be recovered even if there are $\beta=3$ simultaneous failure. In fact,  it is also a $(3,3,0,2)$ robust locally repairable code.

Now, we will use the bound obtained earlier to show that our code is optimal. We have plotted our bound in Figure \ref{fig:Chan_Code_Sim} when 
$n= 16$ and  $\beta=3$. 
The horizontal axis is the code locality (i.e., $r$) and the vertical axis is the dimension of the code (or $\log |\C|$). From the figure, when $r=3$,  the dimension of a $(3,3,1,1)$ robust locally repairable code is at most 9. And our constructed code has exactly 9 dimensions. Therefore, our code meets the bound and is optimal. 
In fact, our bound also indicates that the dimension of a $(3,3,0,2)$ robust locally repairable code is also at most 9.  Therefore, our code is in fact an optimal $(3,3,1,1)$ and $(3,3,0,2)$ robust locally repairable code.

\end{example}

%we derive the numerical results of the optimisation problem in \eqref{Eq:linprog2} to evaluate our code's optimality. The upper bounds in Fig. \ref{fig:Chan_Code_Sim} show the maximum code dimension $\k$ (i.e., maximum codebook size $|\C|=q^{\k}$) for different repair group sizes $r$. The code parameters are corresponding to the code in Example 1 as code length $n=16$, minimum distance $d=4$ (i.e., $\beta=3$ multiple failure recovery), and field size $q=2$ for different number of unavailable survived nodes $\Gamma$ and alternative repair groups $\zeta$. The corresponding point on these two graphs to our code in Example 1 is the point $(r=3,\k=9)$ which means that for the code with these parameters and repairing group size $r \leq 3$, the maximum number of data blocks is $k=9$. Therefore, this follows that the code in Example 1 achieve the upper bounds.  

\begin{figure}[t]
\centering
\includegraphics[width=1\textwidth]{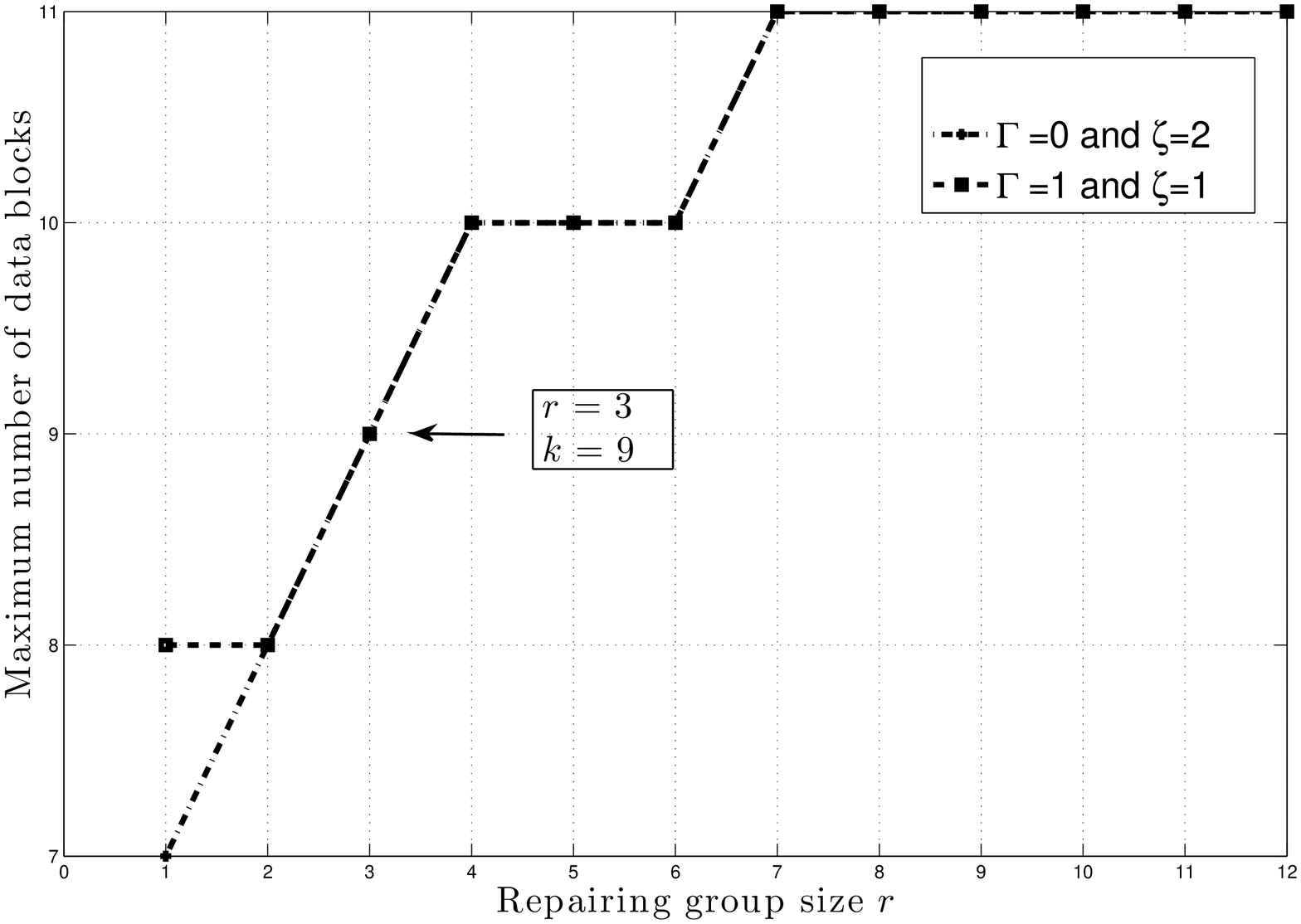}
\caption{Upper bounds for $(3,3,\Gamma,\zeta)$ binary robust locally repairable code of length $n=16$.}
\label{fig:Chan_Code_Sim}
\end{figure}

%%%%%%%%%%%%%%%%%%%%%%%%%%%%%%%
\begin{example}
Consider a binary linear code of length $n=8$ and dimension $k=4$ defined by the following parity check equations:  

\begin{align} \label{Eq:cube1}
P_1 & = C_1 + C_2 + C_3 \\ \label{Eq:cube2}
P_2 & = C_1 + C_2 + C_4 \\ \label{Eq:cube3}
P_3 & = C_2 + C_3 + C_4 \\ \label{Eq:cube4}
P_4 & = C_1 + C_3 + C_4
\end{align}
Again, $C_{1}, \ldots, C_{4} $ are the information bits while $P_{1}, \ldots, P_{4}$ are the parity check bits.
The code can be represented by Figure \ref{fig:Cube_Code}. The above parity check equations imply that 
the sum of any node and its three adjacent nodes in Figure \ref{fig:Cube_Code} is always equal to zero.

This code has a minimum distance of 4. According to the Equations \eqref{Eq:cube1}-\eqref{Eq:cube4}, for every single node failure, there exists 7 repair groups with size $r=3$. 
For example, suppose $C_1$ fails. Then, the repair groups are 

\begin{align*}
&R^{C_1}_1= \{3,4,8\}, R^{C_1}_2 = \{2,4,6\}, R^{C_1}_3 = \{2,3,5\} \\
&R^{C_1}_4 = \{2,7,8\}, R^{C_1}_5 = \{3,6,7\}, R^{C_1}_6 = \{4,5,7\} \\
&R^{C_1}_7 = \{5,6,8\} 
\end{align*}

\begin{figure}[t]
\centering
\includegraphics[width=.35\textwidth]{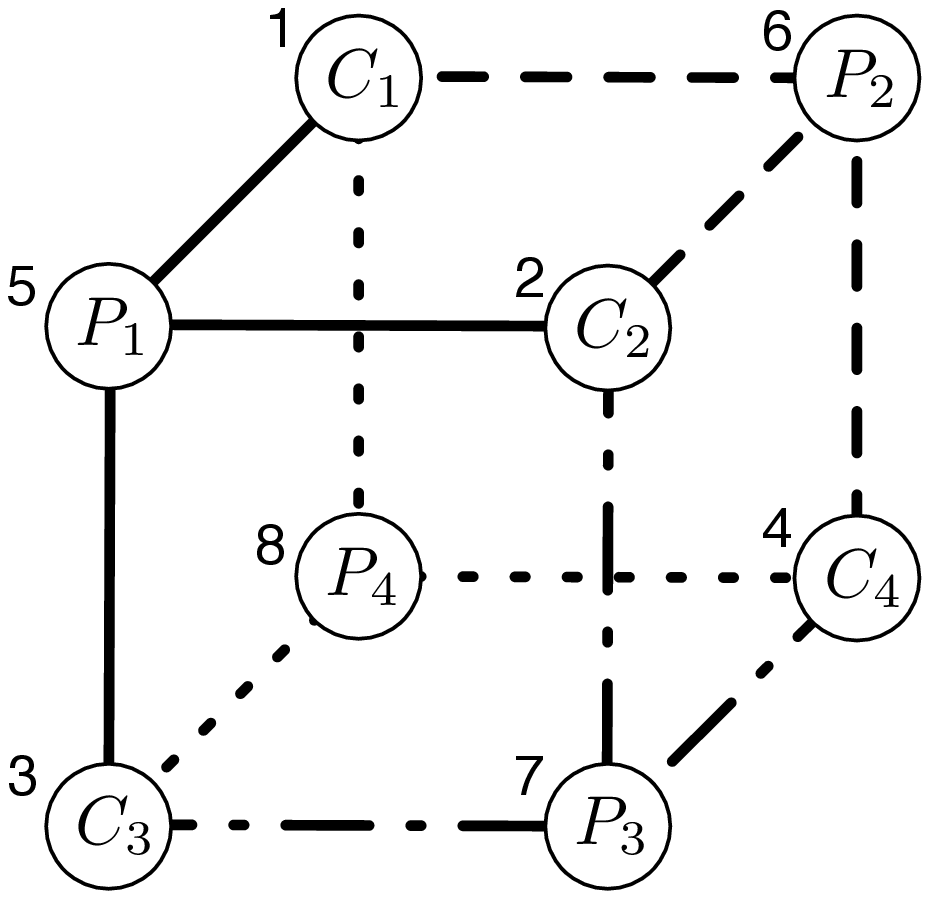}
\caption{A  robust locally repairable code of length $n=8$ and dimension $k=4$ with $\zeta=7$ repair group for any failure.}
\label{fig:Cube_Code}
\end{figure}

Hence, our code is a $(3,3,0,7)$  robust locally repairable code. Furthermore, 
it can be directly verified that for any distinct $i , j \neq 1$, 

\begin{align}\label{eq20}
\left| \left\{  \ell : i \not\in R^{C_{1}}_{\ell} \right\} \right|  = 4
\end{align}
and 

\begin{align}\label{eq21}
\left| \left\{ \ell :  \{i,j\} \cap R^{C_{1}}_{\ell} = \emptyset  \right\} \right| = 2.
\end{align}

Therefore, if any one of the surviving nodes is not available, \eqref{eq20} implies that there are still 4 alternative repair groups. 
This means that our code is also a $(3,3,1,4)$ robust locally repairable code. Similarly, by \eqref{eq21}, our code is also a $(3,3,2,2)$ robust locally repairable code.

As shown in our bound (see Figure \ref{fig:Cube_Code_Sim}), our code  has the highest codebook size for the given parameters  among all
binary $(3,3,0,7)$, $(3,3,1,4)$ and $(3,3,2,2)$ robust locally repairable codes.

%%%%%%%%%%%%%%%%%%%%%%%%Section 7%%%%%%%%%%%%%%%%%%%%%%%%%%%%%%
\section{Conclusion} \label{Sec:conclu}

In this paper, we characterised the coding scheme for robust locally repairable storage codes. This coding scheme overcomes a significant issue of the locally repairable codes which is their repair inefficiency under the circumstances when there exist extra failures or unavailability in the network. This coding scheme guarantees the local repairability of any failure at the presence of extra failures by constructing multiple local repair groups for each node in the network. In this case, if any of the repair groups is not available during the repair process, there exist alternative groups to locally repair the failed node. We also established a linear programming problem to upper bound the size of these codes. This practical bound can optimise the trade-off between different parameters of the code such as minimum distance, code length, locality, the number of alternative repair groups, and the number of the extra failures. We also provided two optimal robust locally repairable code examples.

\begin{figure}[t]
\centering
\includegraphics[width=1.2\textwidth]{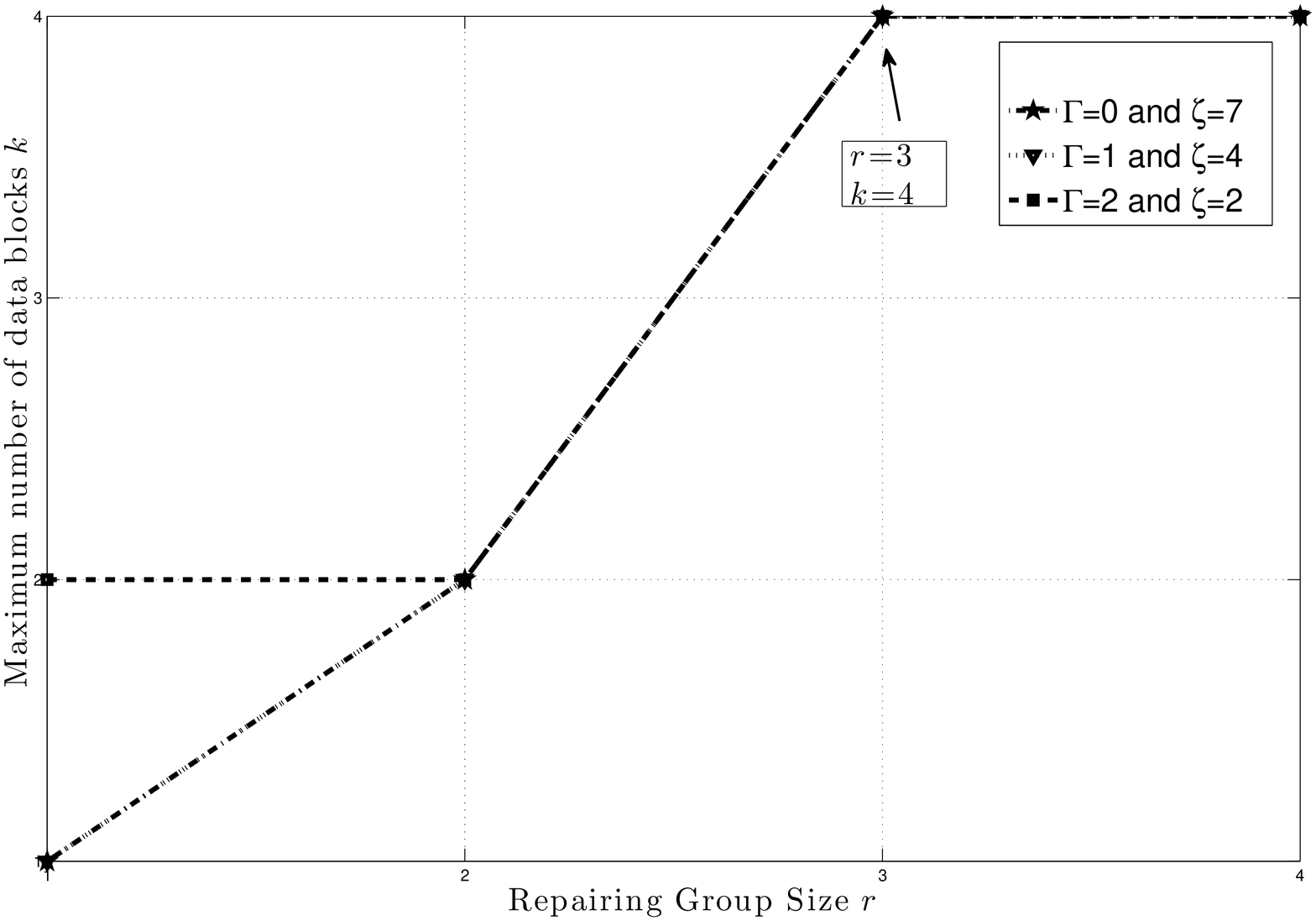}
\caption{Upper bounds for $(3,3,\Gamma,\zeta)$ binary robust locally repairable code of length $n=8$.}
\label{fig:Cube_Code_Sim}
\end{figure}
\end{example}

The update efficiency of the storage networks was addressed. We characterised an update efficient storage code such that any changes in the stored data will result in a small number of node updates. The necessary conditions for existence  of update-efficient storage codes was established.

We also showed how the symmetries in the codes can be exploit to significantly reduce the complexity of the constraints in the linear programming problem and in the necessary conditions.

%\begin{comment}    
%\newpage  
%For acknowledgements section, please don't number the section, please begin it with \section*{Acknowledgements}
%\section*{Acknowledgments} We would like to thank you for \textbf{following
%the instructions above} very closely in advance. It will definitely
%save us lot of time and expedite the process of your paper's
%publication.
%\end{comment}

% You may incorporate your references as follows in your main tex file.
% Using BibTex is not recommended but can be handled.

\bibliographystyle{AIMS}
\bibliography{Ali.bib}

\providecommand{\href}[2]{#2}
\providecommand{\arxiv}[1]{\href{http://arxiv.org/abs/#1}{arXiv:#1}}
\providecommand{\url}[1]{\texttt{#1}}
\providecommand{\urlprefix}{URL }
\begin{thebibliography}{10}

\bibitem{CombinatorialBound}
\newblock A.~Agarwal, A.~Barg, S.~Hu, A.~Mazumdar and I.~Tamo,
\newblock Combinatorial alphabet-dependent bounds for locally recoverable
  codes,
\newblock \emph{IEEE Trans. Inf. Theory}, \textbf{64} (2018), 3481--3492.

\bibitem{ErasVSReplTotalRecall}
\newblock R.~Bhagwan, K.~Tati, S.~S. Y.~Cheng and G.~Voelker,
\newblock Total recall: System support for automated availability management,
\newblock in \emph{Proc. NSDI '04}.

\bibitem{EVENODD}
\newblock M.~Blaum, J.~Brady, J.~Bruck and J.~Menon,
\newblock Evenodd: An efficient scheme for tolerating double disk failures in
  raid architectures,
\newblock \emph{IEEE Trans. Inf. Theory}, \textbf{44} (1995), 192--202.

\bibitem{MDSArray}
\newblock M.~Blaum, J.~Bruck and A.~Vardy,
\newblock \uppercase{MDS} array codes with independent parity symbols,
\newblock \emph{IEEE Trans. Inf. Theory}, \textbf{42} (1996), 529--542.

\bibitem{6478813}
\newblock M.~{Blaum}, J.~L. {Hafner} and S.~{Hetzler},
\newblock Partial-mds codes and their application to raid type of
  architectures,
\newblock \emph{IEEE Transactions on Information Theory}, \textbf{59} (2013),
  4510--4519.

\bibitem{ReplDSN1Farsite}
\newblock W.~J. Bolosky, J.~R. Douceur, D.~ELY and M.~Theimer,
\newblock Feasibility of a serverless distributed file system deployed on an
  existing set of desktop pcs,
\newblock in \emph{Proc. of Sigmetrics}, 2000.

\bibitem{MazumdarBound}
\newblock V.~Cadambe and A.~Mazumdar,
\newblock An upper bound on the size of locally recoverable codes,
\newblock in \emph{proc. Int. Symp. Network Coding (NetCod)},
\newblock Calgary, AB, 2013,
\newblock 1--5.

\bibitem{LRC-size-bound}
\newblock V.~R. Cadambe and A.~Mazumdar,
\newblock Bounds on the size of locally recoverable codes,
\newblock \emph{IEEE Trans. Inf. Theory}, \textbf{61} (2015), 5787--5794.

\bibitem{7740918}
\newblock G.~{Calis} and O.~O. {Koyluoglu},
\newblock A general construction for pmds codes,
\newblock \emph{IEEE Communications Letters}, \textbf{21} (2017), 452--455.

\bibitem{LPLR}
\newblock T.~H. Chan, M.~A. Tebbi and C.~W. Sung,
\newblock Linear programming bounds for storage codes,
\newblock in \emph{9th International Conference on Information, Communication,
  and Signal Processing (ICICS 2013)}, 2013.

\bibitem{Multi_r_delta_local}
\newblock B.~Chen, S.-T. Xia and J.~Hao,
\newblock Locally repairable codes with multiple $(r_i, \delta_i)$-localities,
\newblock in \emph{Proc. IEEE Int. Symp. Information Theory (ISIT)},
\newblock Aachen, Germany, 2017,
\newblock 2038--2042.

\bibitem{ReplvsEras2}
\newblock Y.~Chen, J.~Edler, A.~Goldberg, S.~S. A.~Gottlieb and P.~Yianilos,
\newblock Prototype implementation of archival intermemory,
\newblock in \emph{Proc. of IEEE ICDE}, 1996,
\newblock 485--495.

\bibitem{NCDSS}
\newblock A.~G. Dimakis, P.~B. Godfrey, Y.~Wu, M.~J. Wainwright and
  K.~Ramchandran,
\newblock Network coding for distributed storage systems,
\newblock \emph{IEEE Trans. Inf. Theory}, \textbf{56} (2010), 4539--4551.

\bibitem{ReplDSN2PAST}
\newblock P.~Druschel and A.~Rowstron,
\newblock Storage management and caching in past, a large-scale, persistent
  peer-to-peer storage utility,
\newblock in \emph{Proc. of ACM SOSP}, 2001.

\bibitem{8006478}
\newblock R.~{Gabrys}, E.~{Yaakobi}, M.~{Blaum} and P.~H. {Siegel},
\newblock Constructions of partial mds codes over small fields,
\newblock in \emph{2017 IEEE International Symposium on Information Theory
  (ISIT)}, 2017,
\newblock 1--5.

\bibitem{6846332}
\newblock P.~{Gopalan}, C.~{Huang}, B.~{Jenkins} and S.~{Yekhanin},
\newblock Explicit maximally recoverable codes with locality,
\newblock \emph{IEEE Transactions on Information Theory}, \textbf{60} (2014),
  5245--5256.

\bibitem{LRCGopalan}
\newblock P.~Gopalan, C.~Huang, H.~Simitci and S.~Yekhanin,
\newblock On the locality of codeword symbols,
\newblock \emph{IEEE Trans. Inf. Theory}, \textbf{58} (2012), 6925--6934.

\bibitem{Pyramid}
\newblock C.~Huang, M.~Chen and J.~Li,
\newblock \emph{Pyramid Codes: Flexible Schemes to Trade Space for Access
  Efficiency in Reliable Data Storage Systems},
\newblock Technical Report MSR-TR-2007-25, Microsoft Research, 2007.

\bibitem{UneqLocal}
\newblock S.~Kadhe and A.~Sprintson,
\newblock Codes with unequal locality,
\newblock in \emph{Proc. IEEE Int. Symp. Information Theory (ISIT)},
\newblock Barcelona, Spain, 2016,
\newblock 435--439.

\bibitem{LRCsmallAvail}
\newblock S.~Kadhe and R.~Calderbank,
\newblock Rate optimal binary linear locally repairable codes with small
  availability,
\newblock in \emph{Proc. IEEE Int. Symp. Information Theory},
\newblock Aachen, Germany, 2017,
\newblock 166--170.

\bibitem{6846301}
\newblock G.~M. {Kamath}, N.~{Prakash}, V.~{Lalitha} and P.~V. {Kumar},
\newblock Codes with local regeneration and erasure correction,
\newblock \emph{IEEE Transactions on Information Theory}, \textbf{60} (2014),
  4637--4660.

\bibitem{Khan:2011:SIR:2002218.2002224}
\newblock O.~Khan, R.~Burns, J.~Park and C.~Huang,
\newblock In search of i/o-optimal recovery from disk failures,
\newblock in \emph{Proceedings of the 3rd USENIX Conference on Hot Topics in
  Storage and File Systems},
\newblock HotStorage'11, USENIX Association, Berkeley, CA, USA, 2011,
\newblock 6--6,
\newblock \urlprefix\url{http://dl.acm.org/citation.cfm?id=2002218.2002224}.

\bibitem{HybridOceanstore}
\newblock J.~Kubiatowicz, D.~Bindel, Y.~Chen, S.~Czerwinski, P.~Eaton,
  D.~Geels, R.~Gummadi, S.~Rhea, H.~Weatherspoon, W.~Weimer, C.~Wells and
  B.~Zhao,
\newblock Oceanstore: An architecture for global-scale persistent storage,
\newblock in \emph{Proc. 9th Int. Conf. Architectural Support Programm. Lang.
  Oper. Syst.},
\newblock Boston, MA, 2000,
\newblock 190--201.

\bibitem{macwilliams}
\newblock F.~J. Macwilliams and N.~J.~A. Sloane,
\newblock \emph{The Theory of Error Correcting Codes},
\newblock North Holland, Amsterdam, 1977.

\bibitem{martnez2018universal}
\newblock U.~Martnez-Penas and F.~R. Kschischang,
\newblock Universal and dynamic locally repairable codes with maximal
  recoverability via sum-rank codes,
\newblock in \emph{2018 56th Annual Allerton Conference on Communication,
  Control, and Computing (Allerton)},
\newblock IEEE, 2018,
\newblock 792--799.

\bibitem{SelfRepair}
\newblock F.~Oggier and A.~Datta,
\newblock Self-repairing homomorphic codes for distributed storage systems,
\newblock in \emph{proc. IEEE INFOCOM}, 2011,
\newblock 1251--1223.

\bibitem{LRmulti}
\newblock L.~Pamies-Juarez, H.~D. Hollmann and F.~Oggier,
\newblock Locally repairable codes with multiple repair alternatives,
\newblock in \emph{proc. IEEE Int. Symp. Information Theory}, 2013,
\newblock 892--896.

\bibitem{LRCDimakis}
\newblock D.~S. Papailiopoulos and A.~G. Dimakis,
\newblock Locally repairable codes,
\newblock in \emph{proc. IEEE Int. Symp. Information Theory},
\newblock Cambridge, MA, 2012,
\newblock 2771--2775.

\bibitem{SimpleRegCode}
\newblock D.~Papailiopoulos, J.~Luo, A.~Dimakis, C.~Huang and J.~Li,
\newblock Simple regenerating codes: Network coding for cloud storage,
\newblock in \emph{proc. IEEE INFOCOM}, 2012,
\newblock 2801--2805.

\bibitem{RAID}
\newblock D.~A. Patterson, G.~Gibson and R.~Katz,
\newblock \emph{A case for redundant arrays of inexpensive disks (raid)},
\newblock Tech. Rep. CSD-87-391, Computer Science Division, Department of
  Electrical Engineering and Computer Science, University of California,
  Berkeley, CA 94720, 1987.

\bibitem{LRC2Erasure}
\newblock N.~Prakash, V.~Lalitha and P.~Kumar,
\newblock Codes with locality for two erasures,
\newblock in \emph{Proc. IEEE Int. Symp. Information Theory (ISIT)}, 2014,
\newblock 1962--1966.

\bibitem{6655894}
\newblock A.~S. {Rawat}, O.~O. {Koyluoglu}, N.~{Silberstein} and
  S.~{Vishwanath},
\newblock Optimal locally repairable and secure codes for distributed storage
  systems,
\newblock \emph{IEEE Transactions on Information Theory}, \textbf{60} (2014),
  212--236.

\bibitem{CoopLocal}
\newblock A.~S. Rawat, A.~Mazumdar and S.~Vishwanath,
\newblock On cooperative local repair in distributed storage,
\newblock in \emph{48th Annual Conference on Information Sciences and Systems
  (CISS)}, 2014,
\newblock 1--5.

\bibitem{LRCMultDisjoint}
\newblock A.~S. Rawat, D.~S. Papailiopoulos, A.~G. Dimakis and S.~Vishwanath,
\newblock Locality and availability in distributed storage,
\newblock in \emph{proc. IEEE Int. Symp. Information Theory}, 2014,
\newblock 681--685.

\bibitem{ReplvsErasOcean}
\newblock S.~Rhea, C.~Wells, P.~Eaton, D.~Geels, B.~Zhao, H.~Weatherspoon and
  J.~Kubiatowicz,
\newblock Maintenance free global storage in oceanstore,
\newblock in \emph{Poc. of IEEE Internet Computing}, 2001,
\newblock 40--49.

\bibitem{TamotBound}
\newblock I.~Tamo and A.~Barg,
\newblock Bounds on locally recoverable codes with multiple recovering sets,
\newblock in \emph{proc. IEEE Int. Symp. Information Theory},
\newblock Honolulu, HI, 2014,
\newblock 691--695.

\bibitem{OptimalLocal}
\newblock I.~Tamo and A.~Barg,
\newblock A family of optimal locally recoverable codes,
\newblock \emph{IEEE Trans. Inf. Theory}, \textbf{60} (2014), 4661--4676.

\bibitem{LocalMatroid}
\newblock I.~Tamo, D.~Papailiopoulos and A.~Dimakis,
\newblock Optimal locally repairable codes and connections to matroid theory,
\newblock in \emph{proc. IEEE Int. Symp. Information Theory},
\newblock Istanbul, 2013,
\newblock 1814--1818.

\bibitem{RLLC}
\newblock M.~A. Tebbi, T.~H. Chan and C.~Sung,
\newblock Linear programming bounds for robust locally repairable storage
  codes,
\newblock in \emph{proc. Information Theory Workshop (ITW), 2014 IEEE}, 2014,
\newblock 50--54.

\bibitem{LocComb}
\newblock A.~Wang and Z.~Zhang,
\newblock Repair locality from a combinatorial perspective,
\newblock in \emph{proc. IEEE Int. Symp. Information Theory}, 2014,
\newblock 1972--1976.

\bibitem{ErasVsRepl}
\newblock H.~Weatherspoon and J.~D. Kubiatowicz,
\newblock Erasure coding vs. replication: A quantitative comparison,
\newblock in \emph{Proc. Int. Workshop Peer-to-Peer Syst.}, 2002.

\bibitem{ERC}
\newblock S.~B. Wicker,
\newblock \emph{Error Control Systems for Digital Communication and Storage},
\newblock Prentice Hall, Englewood Cliffs, NJ, 1995.

\bibitem{MultiLocal}
\newblock A.~Zeh and E.~Yaakobi,
\newblock Bounds and constructions of codes with multiple localities,
\newblock in \emph{Proc. IEEE Int. Symp. Information Theory (ISIT)},
\newblock Barcelona, Spain, 2016,
\newblock 640--644.

\end{thebibliography}

\medskip
% The data information below will be filled by AIMS editorial staff
Received xxxx 20xx; revised xxxx 20xx.
\medskip

\end{document}